\documentclass{article}
\usepackage{ijcai22}

\usepackage{times}
\usepackage{soul}
\usepackage{url}
\usepackage[hidelinks]{hyperref}
\usepackage[utf8]{inputenc}
\usepackage[small]{caption}
\usepackage{graphicx}
\usepackage{amsmath}
\usepackage{amsthm}
\usepackage{amsfonts}
\usepackage{booktabs}
\usepackage{tikz}
\usetikzlibrary{matrix}
\usepackage{algorithm}
\usepackage{algorithmic}
\usepackage{subcaption}

\usepackage{dblfloatfix}

\newtheorem*{theorem*}{Theorem}
\newtheorem{theorem}{Theorem}[section]
\newtheorem{lemma}[theorem]{Lemma}

\newtheorem{example}[theorem]{Example}

\newtheorem{definition}[theorem]{Definition}

\newcommand{\mi}{\textsc{MI}}

\newcommand{\E}{\mathrm{E}}

\newcommand{\sign}{\mathrm{sgn}}

\newcommand{\modu}{\mathbf{Mod}}
\newcommand{\rk}{\mathrm{Rank}}

\title{Modularity and Mutual Information in Networks: Two Sides of the Same Coin}

\author{
Qian Wang$^1$\and
Yongkang Guo$^1$\and
Zhihuan Huang$^1$\And
Yuqing Kong$^1$\\
\affiliations
$^1$CFCS, Peking University\\
\emails
\{charlie, yongkang\_guo, zhihuan.huang, yuqing.kong\}@pku,edu,cn
}

\begin{document}

\maketitle

\begin{abstract}

Modularity, first proposed by \cite{newman2004finding}, is one of the most popular ways to quantify the significance of community structure in complex networks. It can serve as both a standard benchmark to compare different community detection algorithms, and an optimization objective to detect communities itself. Previous work on modularity has developed many efficient algorithms for modularity maximization. However, few of researchers considered the interpretation of the modularity function itself. In this paper, we study modularity from an information-theoretical perspective and show that modularity and mutual information in networks are essentially the same. The main contribution is that we develop a family of generalized modularity measures, $f$-modularity based on $f$-mutual information. $f$-Modularity has an information-theoretical interpretation, enjoys the desired properties of mutual information measure, and provides an approach to estimate the mutual information between discrete random variables. At a high level, we show the significance of community structure is equivalent to the amount of information contained in the network. The connection of $f$-modularity and $f$-mutual information bridges two important fields, complex network and information theory and also sheds light on the design of measures on community structure in future.
\end{abstract}

\section{Introduction}
\label{introduction}

Networks have been attracting considerable attention over the past few decades as a representation of real data in many complex system applications, including natural, social, and technological systems. One of the most important characteristics that have been found to occur commonly in these networks is community structure \cite{girvan2002community,fortunato2010community,porter2009communities,malliaros2013clustering,cherifi2019community}. It is a natural idea to partition a complex network into multiple modules or communities by grouping nodes into sets such that each set of nodes are densely connected internally while cross-group connections are sparse. Therefore, research on community structure occupies an important part in data processing and data analyzing.

One of the most widely used tools for analyzing community structure is modularity \cite{newman2004fast}. The modularity function is defined as $$Q=\sum_i (e_{ii}- a_i^2),$$ where $e_{ij}$ denotes the fraction of edges connecting community $i$ to community $j$ (hence $e_{ii}$ for the edges within community $i$) and $a_i=\sum_j e_{ij}$. In essence, modularity measures the ``distance'' of the community structure in the real network from a random network without any community structure, thus a higher modularity implies a clearer community structure. There has been extensive research on modularity-based techniques obtaining maximization algorithms with faster speed or higher accuracy. However, few work dived into the concept of modularity itself.

In this paper, we study modularity from an information-theoretical perspective and show that modularity and mutual information in networks are actually two sides of the same coin. By regarding the adjacency matrix of a network as the joint probability distribution of two discrete random variables, we observe an intuitive relation between mutual information and community structure. For example, if the adjacency matrix is a block diagonal matrix, there is naturally a good community partition with a high modularity for this network, while the value of mutual information is also high for this joint distribution at the same time.

Following this intuition, we start from $f$-mutual information~\cite{kong2019information}, a generalization of Shannon mutual information and derive a family of generalized modularity measures, $f$-modularity. For a given network, when we consider the adjacency matrix as the joint probability matrix of two discrete random variables, maximizing $f$-modularity is equivalent to approximating $f$-mutual information. By substituting different convex functions $f$, we can get different instances of $f$-modularity. Actually, by picking a particular $f$ and adding some constraints, we will show that the original definition of Newman's modularity is a special case of ours. For other commonly used smooth convex functions, $f$-modularity surpasses Newman's modularity due to it being differentiable. 

The main contribution of this paper is proposing $f$-modularity, which implies a strong connection between modularity and mutual information in networks. The key insight is that the significance of community structure equals the amount of information contained in the network. $f$-modularity has an information-theoretical interpretation, enjoys the desired properties of mutual information measure, and provides an approach to estimate the mutual information between discrete random variables (Section \ref{f-modularity}). We validate our theoretical results by experiments in Section \ref{experiments}.

\section{Related Work}
\label{related_work}
Modularity was first proposed by \cite{newman2004finding} as a stop criterion for another community detection algorithm. Then in the same year, \cite{newman2004fast} proposed an alternative community detection approach directly based on modularity maximization. They chose an greedy-based approximation algorithm in order to reduce time overhead. Later, modularity maximization was formally proved to be an NP-complete problem by \cite{brandes2007modularity}. 

Following the seminal work of Newman, many approximate optimization methods for modularity maximization were developed, offering different balances between lower complexity and higher accuracy \cite{cherifi2019community}. Some researchers also aimed to address the shortcomings of modularity by proposing new metrics similar to modularity \cite{muff2005local,haq2019community}. Instead of pursuing a better algorithm or making slight modifications to the original definition of modularity, we starts from $f$-mutual information, entirely another concept, to derive a generalized modularity, which includes Newman's modularity as a special case.

As far as we know, all previous work involving both Newman’s modularity and mutual information was related to \textit{normalized mutual information} (NMF) \cite{danon2005comparing}. NMF takes the partitions of the network as random variables while $f$-mutual information in our context takes the edges of the network as random variables. \cite{danon2005comparing} first used NMF as a standard benchmark to compare different approaches for community detection, so later work on modularity just accepted it as a metric. Compared with them, we build a more solid mathematical connection between modularity and mutual information.


There is a only limit of literature looking into modularity itself and establishing its correspondences to other fields, but with no relevance to mutual information. \cite{zhang2014scalable,newman2016equivalence,veldt2018correlation} showed the equivalence between modularity and the maximum likelihood formulation of the degree-corrected stochastic block models (SBM). \cite{masuda2017random} showed modularity is closely related to Markov stability in the random walk model. \cite{chang2018approximate} discussed the relation between modularity maximization and non-negative matrix factorization (NMF). Recently, \cite{young2018universality} established the universality of the stochastic block model and showed all the problems where we partition the network by maximizing some objective function, are equivalent including modularity maximization.

\section{$f$-Modularity}
\label{f-modularity}
In this section, we will provide the definition of $f$-modularity based on the dual form of $f$-mutual information. With different convex functions $f$ and constraint sets $C$, we can get different instances of $f$-modularity, one of which corresponds to Newman's modularity. Finally, we illustrate an approximation algorithm to optimize $f$-modularity.

\subsection{Frequency Matrix and Random Matrix}

This subsection formally describes our setting for networks. First, we define the frequency matrix of a network as follows.

\begin{definition}[Frequency Matrix $\mathbf{F}$]
\label{frequency matrix}
Given a bipartite multigraph $G=\left<U, V, E\right>$ and its biadjacency matrix $\mathbf{B}$, where $B_{u,v}$ is the number of edges between $u\in U$ and $v\in V$, the frequency matrix is defined as $\mathbf{F}=\mathbf{B}/N$ where $N=\sum_{u, v}B_{u, v}$ is the total number of edges. For a non-bipartite multigraph $G=\left<V, E\right>$, given its adjacency matrix $\mathbf{A}$, we define $\mathbf{F}=\mathbf{A}/N$ where $N=\sum_{u, v\in V}A_{u, v}$.
\end{definition}

\begin{example}[Seller-Buyer]
$U$ is a set of sellers and $V$ is a set of buyers. For $N$ contracts, $B_{u,v}$ is the number of times that seller $u$ trades with buyer $v$ and $F_{u, v} = B_{u, v}/N$. 
\end{example}

Note that our work can be applied to both bipartite networks and non-bipartite networks, but we will mainly focus on bipartite networks since any adjacency matrix of a non-bipartite multigraph can be induced to the biadjacency matrix of a bipartite multigraph (see Figure~\ref{induction} for illustration). In the following context, a ``network'', if not particularly indicated, refers to a bipartite multigraph.

\begin{figure}[htpb]
    \centering
    \subcaptionbox{\label{nonbipartite}}{
    \begin{tikzpicture}
            \node[circle, draw, inner sep=0cm, minimum size=.5cm] (v1) at (-0.5, 1.3) {$1$};            
            \node[circle, draw, inner sep=0cm, minimum size=.5cm] (v2) at (0.5, 1.3) {$2$};
            \node[circle, draw, inner sep=0cm, minimum size=.5cm] (v3) at (1, 0) {$3$};
            \node[circle, draw, inner sep=0cm, minimum size=.5cm] (v4) at (0, -1.3) {$4$};
            \node[circle, draw, inner sep=0cm, minimum size=.5cm] (v5) at (-1, 0) {$5$};
            \draw[thick] (v1) to (v2);
            \draw[thick] (v2) to (v3);
            \draw[thick] (v1) to (v3);
            \draw[thick] (v2) to (v4);
            \draw[thick] (v3) to (v5);
            \draw[thick] (v4) to (v5);
    \end{tikzpicture}
    }
    \subcaptionbox{\label{bipartite}}{
    \begin{tikzpicture}
            \node[circle, draw, inner sep=0cm, minimum size=.5cm] (u1) at (-0.75, 1.3) {$u_1$};            
            \node[circle, draw, inner sep=0cm, minimum size=.5cm] (u2) at (-0.75, 0.65) {$u_2$};
            \node[circle, draw, inner sep=0cm, minimum size=.5cm] (u3) at (-0.75, 0) {$u_3$};
            \node[circle, draw, inner sep=0cm, minimum size=.5cm] (u4) at (-0.75, -0.65) {$u_4$};
            \node[circle, draw, inner sep=0cm, minimum size=.5cm] (u5) at (-0.75, -1.3) {$u_5$};
            \node[circle, draw, inner sep=0cm, minimum size=.5cm] (v1) at (0.75, 1.3) {$v_1$};            
            \node[circle, draw, inner sep=0cm, minimum size=.5cm] (v2) at (0.75, 0.65) {$v_2$};
            \node[circle, draw, inner sep=0cm, minimum size=.5cm] (v3) at (0.75, 0) {$v_3$};
            \node[circle, draw, inner sep=0cm, minimum size=.5cm] (v4) at (0.75, -0.65) {$v_4$};
            \node[circle, draw, inner sep=0cm, minimum size=.5cm] (v5) at (0.75, -1.3) {$v_5$};
            \draw[thick] (u1) to (v2);
            \draw[thick] (u1) to (v3);
            \draw[thick] (u2) to (v1);
            \draw[thick] (u2) to (v3);
            \draw[thick] (u2) to (v4);
            \draw[thick] (u3) to (v1);
            \draw[thick] (u3) to (v2);
            \draw[thick] (u3) to (v5);
            \draw[thick] (u4) to (v2);
            \draw[thick] (u4) to (v5);
            \draw[thick] (u5) to (v3);
            \draw[thick] (u5) to (v4);
    \end{tikzpicture}
    }
    \subcaptionbox{\label{egmatrix}}{
    $\begin{pmatrix}
    0 & 1 & 1 & 0 & 0 \\
    1 & 0 & 1 & 1 & 0 \\
    1 & 1 & 0 & 0 & 1 \\
    0 & 1 & 0 & 0 & 1 \\
    0 & 0 & 1 & 1 & 0
    \end{pmatrix}$
    }
    \caption{The adjacency matrix of the non-bipartite graph in (\subref{nonbipartite}) is exactly the same with the biadjacency matrix of the bipartite graph in (\subref{bipartite}), i.e., the matrix in (\subref{egmatrix}). As the biadjacency matrix is symmetric in this situation, any rational algorithm should deliver a symmetric partition result for the bipartite graph (\subref{bipartite}). For example, if the communities turn out to be $\{u_1, u_2, u_3, v_1, v_2, v_3\}$ and $\{u_4, u_5, v_4, v_5\}$ in the bipartite graph (\subref{bipartite}), then we know the communities in the non-bipartite graph (\subref{nonbipartite}) is $\{1, 2, 3\}$ and $\{4, 5\}$. Therefore, we only need to consider bipartite networks. }
    \label{induction}
\end{figure}
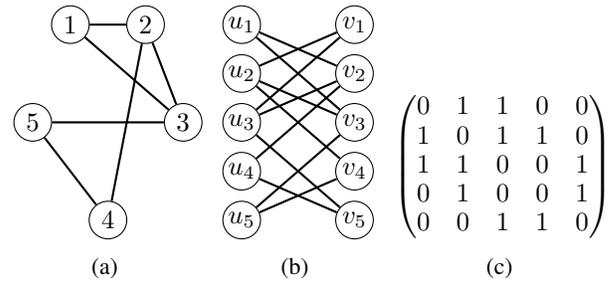

Next, we will define the random matrix (\textit{aka} the null model). At a high level, modularity is a measure for the distance from the real network to the random network.

\begin{definition}[Random Matrix $\mathbf{J}$]
\label{random matrix}
Given a bipartite multigraph $G=\left<U, V, E\right>$ and its frequency matrix $\mathbf{F}$, the random matrix $\mathbf{J}$ is defined as $$J_{u,v}=(\deg(u)\deg(v)-\frac{1}{N}F_{u,v})\frac{N}{N-1},$$ where $\deg(u)=\sum_{j\in V} F_{u,j}$ is the normalized degree of vertex $u$ and similarly, $\deg(v)=\sum_{i\in U} F_{i,v}$.
\end{definition}

If the $N$ edges of $\mathbf{G}$ is $N$ i.i.d. samples of random variables $(X,Y)$ whose realizations are in $(U,V)$, we have
\begin{align*}
\E[F_{u,v}]\quad&=\Pr[X=u,Y=v],\\
\E[\deg(u)]&=\Pr[X=u], \\
\E[\deg(v)]&= \Pr[Y=v].
\end{align*}
The following Lemma \ref{lem:mar} tells us, $\mathbf{J}$ in Definition \ref{random matrix} is an unbiased estimation of $\Pr[X=u]\Pr[Y=v]$.

\begin{lemma}\label{lem:mar}
$\E[J_{u,v}]=\Pr[X=u]\Pr[Y=v]$
\end{lemma}

\begin{proof}[Proof of Lemma~\ref{lem:mar}]
Let 
\[
   \mathrm{I}_n(u, v)=\left\{\begin{array}{ll}
        1 & \text{if the }n\text{-th edge is }(u, v), \\
        0 & \text{otherwise}.
    \end{array}\right.
\]
Then
\begin{align*}
    &\sum_i F_{i,v} \sum_j F_{u,j} \\
    =& \frac{1}{N^2} \left(\sum_{i,n} \mathrm{I}_{n}(i,v)\right) \left(\sum_{j,n} \mathrm{I}_n(u,j)\right) \\
    =& \frac{1}{N^2} \left( \sum_n \mathrm{I}_n(u,v) + \sum_{n\neq n'} \sum_{i,j} \mathrm{I}_{n}(i,v) \mathrm{I}_{n'}(u,j)  \right).
\end{align*}
We get $$\E[\deg(u)\deg(v)] = \frac{1}{N} \Pr(u,v) + \frac{N-1}{N} \Pr(u) \Pr(v), $$ thus $\E[J_{u,v}]=\Pr[X=u]\Pr[Y=v]$.
\end{proof}


\subsection{$f$-Divergence and $f$-Mutual Information}
This subsection introduces $f$-divergence, Fenchel's duality and $f$-mutual information. These are the main technical ingredients for defining $f$-modularity.


\begin{definition}[$f$-Divergence \cite{ali1966general}]\label{def:df}
$f$-Divergence $d_f:\Delta_{\Sigma}\times \Delta_{\Sigma}\mapsto \mathbb{R}$ is a non-symmetric measure of the difference between distribution $\mathbf{p}\in \Delta_{\Sigma} $ and distribution $\mathbf{q}\in \Delta_{\Sigma} $
and is defined to be $$d_f(\mathbf{p};\mathbf{q})=\sum_{\sigma\in \Sigma}
\mathbf{q}(\sigma)f\bigg( \frac{\mathbf{p}(\sigma)}{\mathbf{q}(\sigma)}\bigg),$$
where $f:\mathbb{R}\mapsto\mathbb{R}$ is a convex function and $f(1)=0$.
\end{definition}

As an example, by picking $f(t)=t\log(t)$, we get KL-divergence $d_{KL}(\mathbf{p},\mathbf{q})=\sum_{\sigma}\mathbf{p}(\sigma)\log\frac{\mathbf{p}(\sigma)}{\mathbf{q}(\sigma)}$.


\begin{definition}[Fenchel Duality \cite{rockafellar1966extension}]\label{def:dual}
Given any function $f:\mathbb{R}\mapsto \mathbb{R}$, we define its convex conjugate $f^{\star}$ as a function that also maps $\mathbb{R}$ to $\mathbb{R}$ such that $$f^{\star}(x)=\sup_{t} tx-f(t).$$
\end{definition}

\begin{lemma}[Dual Form of $f$-Divergence\cite{nguyen2010estimating}]\label{lemma:dual}
\begin{align*}
    d_f(\mathbf{p};\mathbf{q}) \geq & \sup_{u \in \mathcal{U}} \sum_{\sigma\in \Sigma}u(\sigma) \mathbf{p}(\sigma)- \sum_{\sigma\in \Sigma}f^{\star}(u(\sigma))\mathbf{q}(\sigma) \\
    = & \sup_{u \in \mathcal{U}} \E_{\mathbf{p}} u- \E_{\mathbf{q}}f^{\star}(u),
\end{align*}
where $\mathcal{U}$ is a set of functions that maps $\Sigma$ to $\mathbb{R}$. The equality holds if and only if $u(\sigma)\in \partial{f}(\frac{\mathbf{p}(\sigma)}{\mathbf{q}(\sigma)})$, i.e., the subdifferential of $f$ on value $\frac{\mathbf{p}(\sigma)}{\mathbf{q}(\sigma)}$.  
\end{lemma}

Function $u$ is a distinguisher between distribution $\mathbf{p}$ and distribution $\mathbf{q}$ and the best distinguisher $u^*$ (if not restricted by $\mathcal{U}$) maximizes the right side to be $d_f(\mathbf{p};\mathbf{q})$. With the above lemma, we can also write the dual form as 
\[
d_f(\mathbf{p};\mathbf{q}) = \sup_{D\in \mathcal{D}} \E_{\mathbf{p}} \partial{f}(D)- \E_{\mathbf{q}}f^{\star}(\partial{f}(D))
\]
where $\mathcal{D}$ is a set of functions that maps $\Sigma$ to $\mathbb{R}$ and the best $D^*$ satisfies $D^*(\sigma)=\frac{\mathbf{p}(\sigma)}{\mathbf{q}(\sigma)}$. Some common $f$ functions for $f$-divergence and their dual forms are shown in Table~\ref{table:distinguishers}.

\begin{table*}[htp]
\caption{Reference Table for $f,f^{\star}$ \protect\cite{kong2019information} }
\begin{center}
\begin{tabular}{llll}
\toprule
    {$f$-Divergence} & {$f(t)$} & {$\partial{f}(D)$} & {$f^{\star}(\partial{f}(D)$)} \\ \midrule\midrule
    Total Variation Distance  & $|t-1|$  & $\sign(\log D)$ & $\sign(\log D)$ \\ 
    \midrule
    KL-Divergence & $t\log t$  & $\log D + 1$ & $D$ \\ 
    \midrule
    Pearson $\chi^2$ & $(t-1)^2$  & $2(D-1)$ & $D^2-1$ \\ 
    \midrule
    Jensen-Shannon & $-(t+1)\log{\frac{t+1}{2}}+t\log t$  & $\log{\frac{2D}{1+D}}$ & $-\log(\frac{2}{1+D})$ \\
    \midrule
    Squared Hellinger &$(\sqrt{t}-1)^2$ & $1-\sqrt{\frac{1}{D}}$ & $\sqrt{D}-1$ \\
    \bottomrule
\end{tabular}
\end{center}
\label{table:distinguishers}
\end{table*}

\begin{definition}[$f$-Mutual Information \cite{kong2019information}]\label{def:mif}
Given two random variables $X,Y$, the $f$-mutual information between $X$ and $Y$ is defined as
\begin{align*}
    \mi^f(X;Y)=&d_f(XY;X\otimes Y)\\
    =&\sum_{x,y}\Pr(x)\Pr(y)f\left(\frac{\Pr(x,y)}{\Pr(x)\Pr(y)}\right)
\end{align*} where $d_f$ is $f$-divergence. 
\end{definition}

$f$-Mutual information measures the correlation of two random variables $X$ and $Y$ via $f$-divergence between the joint distribution, denoted as $XY$, and the product of the marginal distributions, denoted as $X\otimes Y$. As an example, by picking $f$-divergence as KL-divergence, i.e., $f(t)=t\log t$  \cite{cover1999elements}, we get the classic Shannon mutual information,
\[
    \mi^{Shannon}(X;Y)=\sum_{x,y}\Pr(x, y)\log \left(\frac{\Pr(x,y)}{\Pr(x)\Pr(y)}\right).
\]

\begin{lemma}[Properties of $f$-Mutual Information \cite{kong2019information}]
\label{properties}
$f$-Mutual information satisfies
\begin{description}
\item [Symmetry:] $\mi^f(X;Y)=\mi^f(Y;X)$;
\item [Non-negativity:] $\mi^f(X;Y)$ is always non-negative and is 0 if $X$ is independent of $Y$;
\item [Information Monotonicity:] $\mi^f(T(X);Y)\leq \mi^f(X;Y)$ where $T(\cdot)\in \mathbb{R}^{|\Sigma_X|\times |\Sigma_X|}$ is a possibly random operator on $X$ whose randomness is independent of $Y$.
\end{description}
\end{lemma}

\subsection{$f$-Modularity}
\label{definefmodularity}
By Lemma \ref{lemma:dual}, we have the dual form of $f$-Mutual information, 
\begin{align*}
\mi^f(X;Y)= \sup_D \E_{XY} \partial{f}(D) -\E_{X\otimes Y}f^*(\partial{f}(D))\\
= \sum_{u,v} \partial{f}(D_{u,v})\Pr(u, v) - f^*\left(\partial{f}(D_{u,v})\right) \Pr(u)\Pr(v),
\end{align*}
which inspires the definition of $f$-modularity.

\begin{definition}[$f$-Modularity]
Given a bipartite mutligraph $\mathbf{G}$ whose frequency matrix is $\mathbf{F}\in[0,1]^{|U|\times |V|}$, with a constraint set C, the $f$-modularity of $\mathbf{G}$ is defined as 
\[ \modu^f(\mathbf{G})=\max_{\mathbf{D}\in C} \sum_{u,v}\left[\partial{f}(D_{u,v}) F_{u,v}- f^{\star} (\partial{f}(D_{u,v}))J_{u,v} \right].\]
\end{definition}

The matrix $\mathbf{D}$ is a distinguisher that aims to separate the frequency matrix $\mathbf{F}$ and the random matrix $\mathbf{J}$. Thus, $f$-modularity quantifies not only the amount of information in the graph, but also the concept of community structure by measuring the statistical distance between $\mathbf{F}$ and $\mathbf{J}$. Meanwhile, it inherits the information-theoretical properties from $f$-mutual information, validated in Section \ref{experiments}.

Note we regard the real network $\mathbf{G}$ as a noisy realization of the underlying joint distribution, so the constraint set $C$ in the above definition controls the robustness of $f$-modularity. If $C$ is too rich, the robustness will be hurt and if $C$ is too restricted, $\mathbf{F}$ and $\mathbf{J}$ may not be separated properly.

\paragraph{Instances of $f$-Modularity.} We provide several special instances of $f$-modularity by picking different $f$ from Table \ref{table:distinguishers}.

\begin{example}[TVD-Modularity]
\label{tvd}
When $f(t)=|t-1|$ and replacing $\sign(\log(D_{u,v}))$ by $S_{u,v}\in\{-1,1\}$, we obtain
\[ 
\modu^{tvd}(\mathbf{G})=\max_{\mathbf{S}\in\{-1,1\}^{|U|\times|V|}} \sum_{u,v} S_{u,v}\left( F_{u,v}- J_{u,v} \right).
\]
\end{example}

\begin{example}[KL-Modularity]
When $f(t)=t\log t$, 
\[
\modu^{KL}(\mathbf{G})=\max_{\mathbf{D}\in C } \sum_{u,v}\left((\log(D_{u,v})+1) F_{u,v}-D_{u,v}J_{u,v} \right).
\]
\end{example}

\begin{example}[Pearson-Modularity]
\label{pearson}
When $f(t)=(t-1)^2$, 
\[
\modu^{Pearson}(\mathbf{G})= -1 + \max_{\mathbf{D}\in C} \sum_{u,v} \left(2 D_{u,v} F_{u,v} - D_{u,v}^2J_{u,v} \right).
\]
\end{example}

Now we consider the low-rank constraint $\rk(\mathbf{D})\leq r$, which comes naturally when someone is going to restrict a matrix composed of real-world data. With the constraint set $C=\{\mathbf{D}=\mathbf{P} \mathbf{Q}^{\top}|\mathbf{P}\in\mathbb{R}^{|U|\times r},\mathbf{Q}\in\mathbb{R}^{|V|\times r}\}$, the definition can be rewritten as
\begin{align*}
    \modu^f(\mathbf{G})=\max_{\mathbf{P},\mathbf{Q}} \sum_{u,v} & \left[ \partial{f}(\mathbf{P}_{u}\mathbf{Q}^{\top}_{v}) F_{u,v} \right. \\
     & \left. - f^{\star} (\partial{f}(\mathbf{P}_{u}\mathbf{Q}^{\top}_{v}))J_{u,v} \right].
\end{align*}
where $\mathbf{P}_{u}$ denotes the $u^{th}$ row of matrix $\mathbf{P}$ (similar for $\mathbf{Q}_{v}$). As for a non-bipartite network, the constraint can be picked as $C=\{\mathbf{D}=\mathbf{P} \mathbf{P}^{\top}|\mathbf{P}\in\mathbb{R}^{|U|\times r}\}$ due to the symmetry. 

With the $r$-rank constraint, maximizing $f$-modularity in fact finds the optimal embedding of vertices in $U$ and $V$ in $r$-dimensional space $\mathbb{R}^r$, $\mathbf{P}^*$ and $\mathbf{Q}^*$, which is actually equivalent to detecting $r$ communities with overlapping, where one vertex can belong to more than one community. Interestingly, with a special constraint, TVD-modularity will exactly lead to the original definition of modularity \cite{newman2006finding} and output a fairly good partition without overlapping. 

\paragraph{Newman's modularity $\approx$ TVD-Modularity.} For a non-biapartite network $G=\left<V, E\right>$, we denote $g_v$ to be the community to which vertex $v\in V$ belongs. With the notations of our model, Newman's modularity can be written as
$$Q=\sum_{u, v}(F_{u,v}-J_{u, v}')\delta(g_u, g_v),$$
where $J_{u, v}' = \deg(u)\deg(v)$, and 
\[
    \delta(a, b)=\left\{\begin{array}{ll}
        1 & a=b, \\
        0 & \text{otherwise}.
    \end{array}\right.
\]
Although Newman used a biased estimation of $\Pr[X=u, Y=v]$, notice that $\lim_{N\to \infty}J_{u,v}' = J_{u, v}$.

Let us consider the division of a non-bipartite network into just two communities as the greedy algorithm described in \cite{newman2004fast,newman2006modularity}. Let the index vector $\mathbf{s}$ be 
\[
    s_v=\left\{\begin{array}{ll}
        1 & \text{if vertex }v\text{ belongs to community }1, \\
        0 & \text{if vertex }v\text{ belongs to community }2,
    \end{array}\right.
\]
We then write Newman's modularity in the form
\[
Q=\sum_{u, v}(F_{u,v}-J_{u, v}')s_us_v.
\]

Given a non-bipartite network, if we pick $C=\{\mathbf{S}=\mathbf{s} \mathbf{s}^{\top}, \mathbf{s}\in \{-1,1\}^{|V|\times 1}\}$, TVD-modularity becomes 
\[
\modu^{tvd}(\mathbf{G}):=\max_{\mathbf{s}\in \{-1,1\}^{|V|\times 1}} \mathbf{s}^{\top} (\mathbf{F}-\mathbf{J})\mathbf{s} 
\]
which is equivalent to the above $Q$. For the division of a network into more than two communities, we just need to relax the 1-rank constraint to a $r$-rank version ($r>1$).

\paragraph{Optimization for smooth functions.} One of the disadvantages of Newman's modularity is optimizing it requires us to solve in a discrete space, $\mathbf{s}\in \{-1,1\}^{|V|}$. Our $f$-modularity can avoid this obstacle by choosing smooth convex functions $f$ so that we can take advantage of the differentiability of $f$-modularity for optimization. Taking Pearson-modularity with the low-rank constraint as an example,
\begin{align*}
\modu^{Pearson}(\mathbf{G}) = -1 +\max_{\mathbf{D}\in C } \sum_{u,v}\left(2 D_{u,v} F_{u,v}- D_{u,v}^2 J_{u,v} \right)&\\
= -1 + \max_{\mathbf{D}\in C}\sum_{u,v}-J_{u,v}(D_{u,v}-\frac{F_{u,v}}{J_{u,v}})^2+J_{u,v}(\frac{F_{u,v}}{J_{u,v}})^2&.
\end{align*}
We can see quantifying modularity is induced to a weighted low rank approximation problem, on which there exists many mature algorithms \cite{srebro2003weighted}. 

Here we propose an approximation algorithm (Algorithm~\ref{algo}) based on an efficient low-rank approximation subroutine, which works for a family of smooth convex functions $f$ under low-rank constraints. Lemma \ref{lemma:dual} shows that if without any constraint, the optimal distinguisher $\mathbf{D}^*$ satisfies that for all $u,v$, $D^*_{u,v}=\frac{F_{u,v}}{J_{u,v}}$ (Line~\ref{optimalD}). Thus, if there is a low-rank constraint, we can optimize $f$-modularity by finding a distinguisher $\mathbf{D}_r$, a low-rank approximation of $\mathbf{D}^*$ (Line~\ref{approximateD}). This can be done well by singular value decomposition (SVD) \cite{golub1971singular}, or non-negative matrix factorization (NMF) \cite{lee1999learning} if non-negativity is required. Rank selection in Line~\ref{selectr} is determined by a threshold $\theta$ such that $r$ is the minimum value satisfying $$\frac{||\mathbf{D}^*-\mathbf{D}_k||_F^2}{||\mathbf{D}^*||_F^2}<\theta$$ where $||\cdot||_F$ is the Frobenius norm. Finally we use $\mathbf{D}_r$ to compute $f$-modularity (Line~\ref{getM}). 

\begin{algorithm}[htpb]
\caption{Approximation of $f$-Modularity}
\label{algo}
\textbf{Input}: Frequency Matrix $\mathbf{F}$\\
\textbf{Parameter}: Threshold $\theta$\\
\textbf{Output}: Modularity $\modu$
\begin{algorithmic}[1] 
\STATE $J_{u,v}=(\deg(u)\deg(v)-\frac{1}{N}F_{u,v})\frac{N}{N-1} $
\FOR{$u\in U, v\in V$}
\STATE\label{optimalD} $D^*_{u,v}=F_{u,v}/max(J_{u,v},\epsilon) \text{ // small } \epsilon \text{ to avoid overflow} $
\ENDFOR
\STATE\label{selectr} $r=Rank\_Selection\left(D^*\right)$
\STATE\label{approximateD} $\mathbf{D}_r=Approximation\left(\mathbf{D}^*,r\right)$
\STATE\label{getM} $\modu=\sum_{u,v}\left(\partial{f}(D_{u,v}) F_{u,v}- f^{\star} (\partial{f}(D_{u,v}))J_{u,v} \right)$
\STATE \textbf{return} $\modu$
\end{algorithmic}
\end{algorithm}

We emphasize that the main objective for Algorithm~\ref{algo} is to show $f$-modularity with a smooth function $f$ can be optimized in an easier way than the original modularity, rather than investigating the bound of approximation ratio or efficiency potentials on real-world data if actually implemented in practice.

\section{Numerical Experiments}
\label{experiments}
In this section, we validate the inherited information-theoretic properties of $f$-modularity on synthetic data. We generate bipartite multigraphs from known distributions and show that $f$-modularity 1) vanishes when there is no community structure, 2) decreases as communities are contracted, and 3) estimates $f$-mutual information well.
    
\subsection{Data Generation}

We use the stochastic block model~\cite{holland1983stochastic} as the generator. First, we divide two sets of vertices $(U, V)$ into $m$ communities, $\{(U_1, V_1), (U_2, V_2), \cdots, (U_m, V_m)\}$, each with $2n$ vertices and $|U_i|=|V_i|=n, \forall i\in \{1, 2, \cdots, m\}$. We generate $G=\left<U, V, E\right>$ with the probability of edge $(u, v)$ to be
\[
    p(u,v)\propto
    \begin{cases}
        1 & $u,v$\text{ are in the same community},\\
        \alpha & \text{otherwise},
    \end{cases}
\]
where $\alpha \in [0,1]$ is a hyper-parameter. See Figure~\ref{fig:gen_example} for an example. The sum of probability over all edges is normalized to~$1$ so that the probability matrix $\mathbf{P}$ can be regarded as a joint probability distribution between two random variables $X$ and $Y$ with $\Pr[X=u,Y=v] = p(u, v)$.

\begin{figure}[htpb]
    \centering
    \begin{subfigure}{0.4\columnwidth}
        \includegraphics[width=\columnwidth]{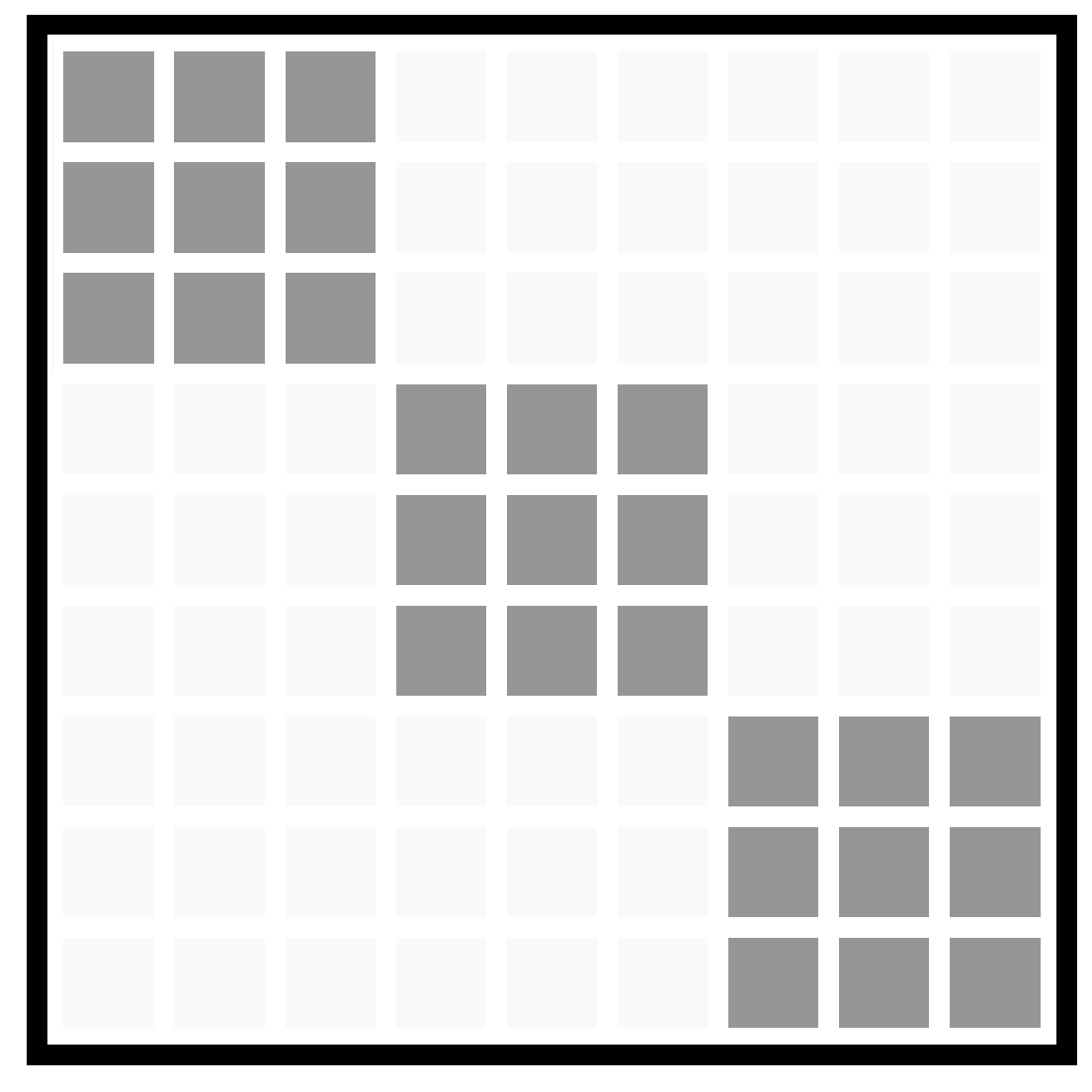}
    \end{subfigure}
    \begin{subfigure}{0.4\columnwidth}
        \includegraphics[width=\columnwidth]{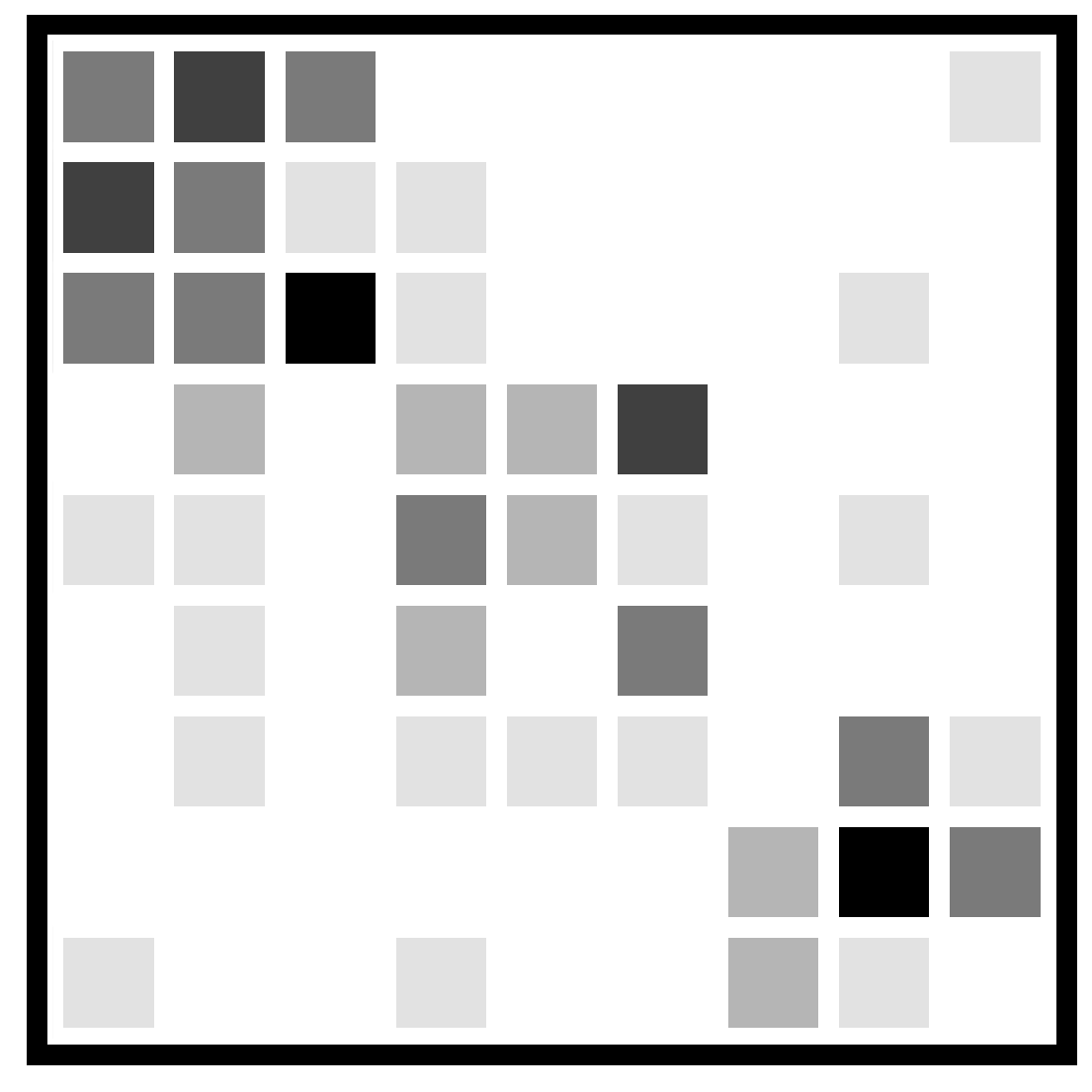}
    \end{subfigure}
    \caption{\textbf{Stochastic block model.} We set the number of communities $m=3$, the size of each community $n=3$ and $\alpha=0.1$. The left matrix represents the underlying probability distribution and the right one is the frequency matrix of a realization. }
    \label{fig:gen_example}
\end{figure}

\begin{figure}[htpb]
    \centering
    \begin{subfigure}{0.4\columnwidth}
        \includegraphics[width=1\textwidth]{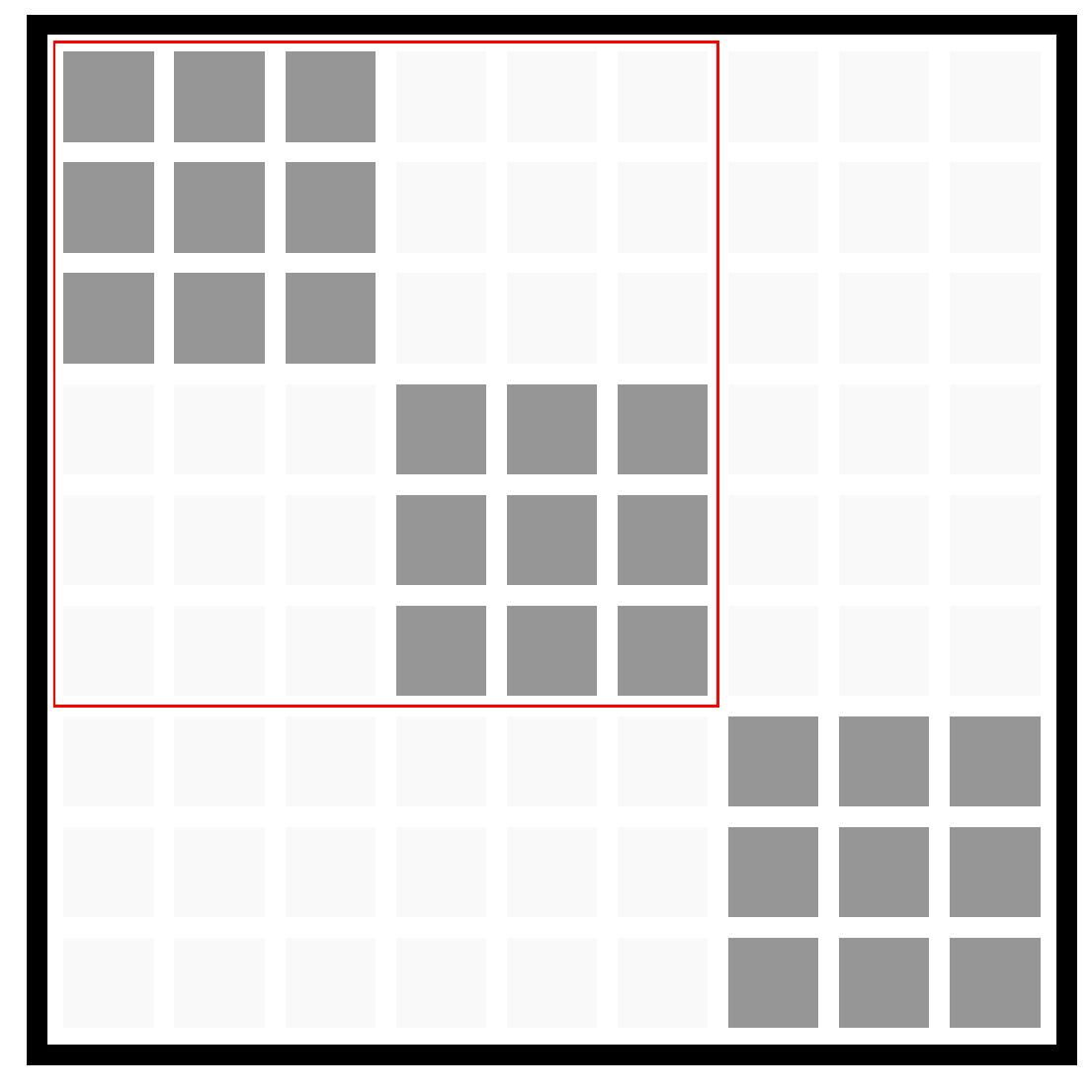}
    \end{subfigure}
    \begin{subfigure}{0.4\columnwidth}
        \includegraphics[width=1\textwidth]{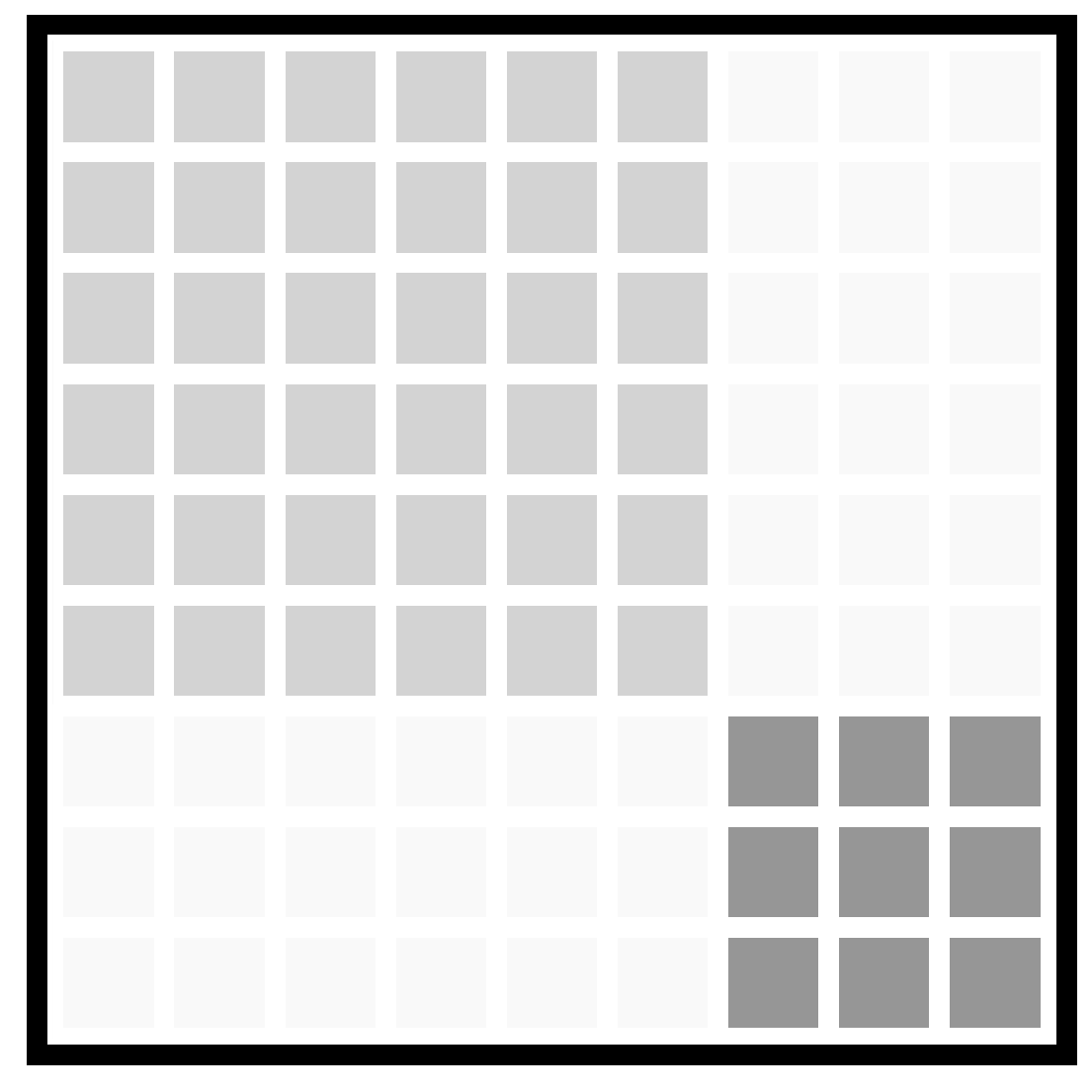}
    \end{subfigure}
    \caption{\textbf{Community contraction process.} We contract the first two communities in the left network to obtain the right one.}
    \label{fig:block_example}
\end{figure}

\begin{figure*}[htpb]
    \subcaptionbox{$\alpha$=0.1}{
            \begin{minipage}{\columnwidth}
                \includegraphics[width=\textwidth]{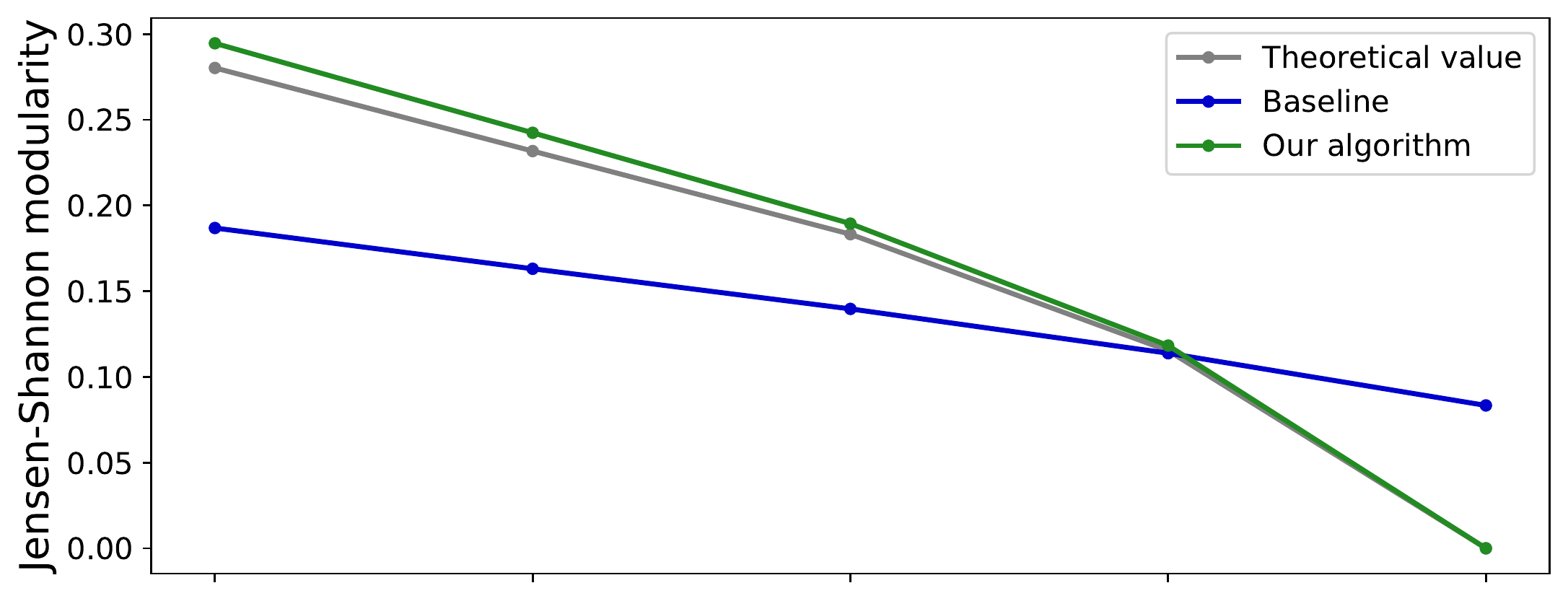}\\
                \null\hspace{0.045\textwidth}
                \includegraphics[width=0.18\textwidth]{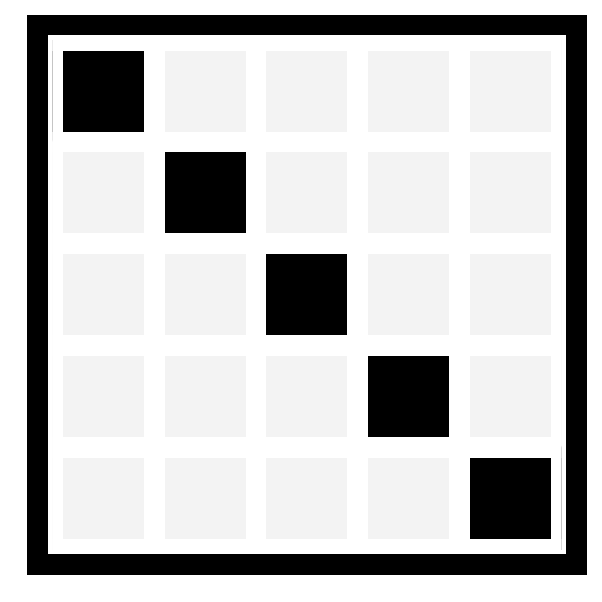}
                \includegraphics[width=0.18\textwidth]{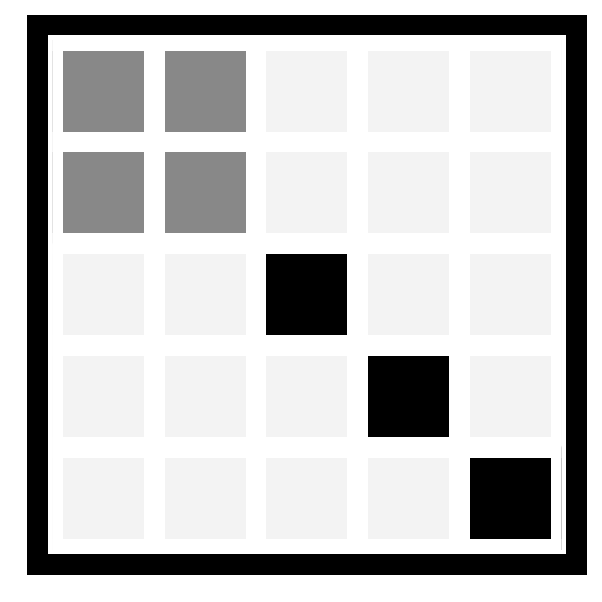}
                \includegraphics[width=0.18\textwidth]{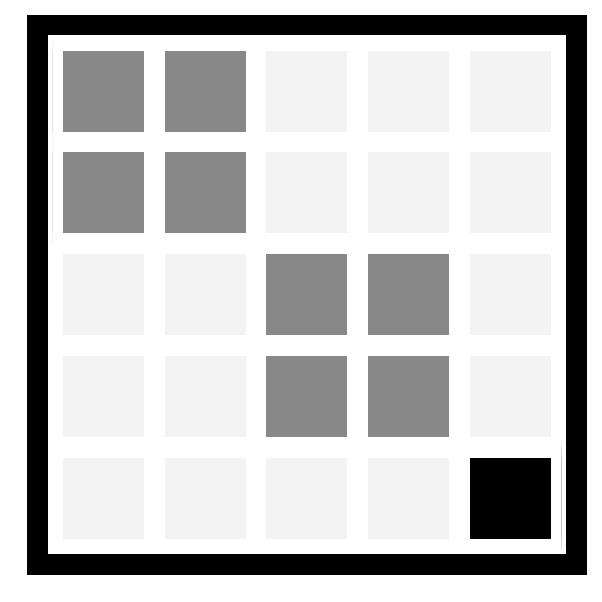}
                \includegraphics[width=0.18\textwidth]{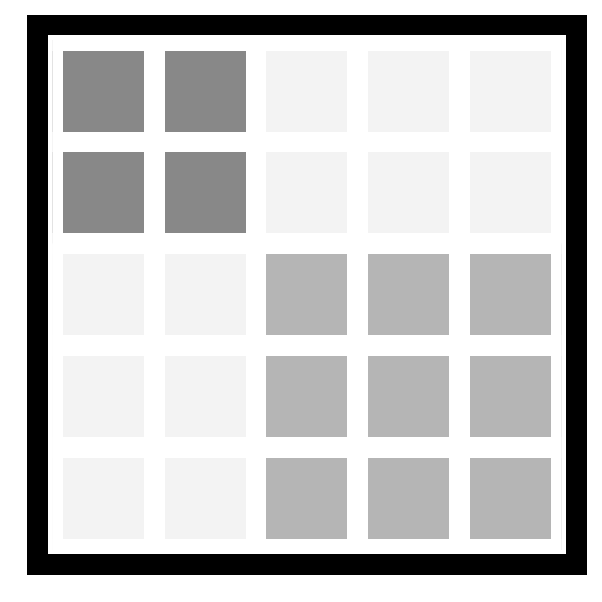}
                \includegraphics[width=0.18\textwidth]{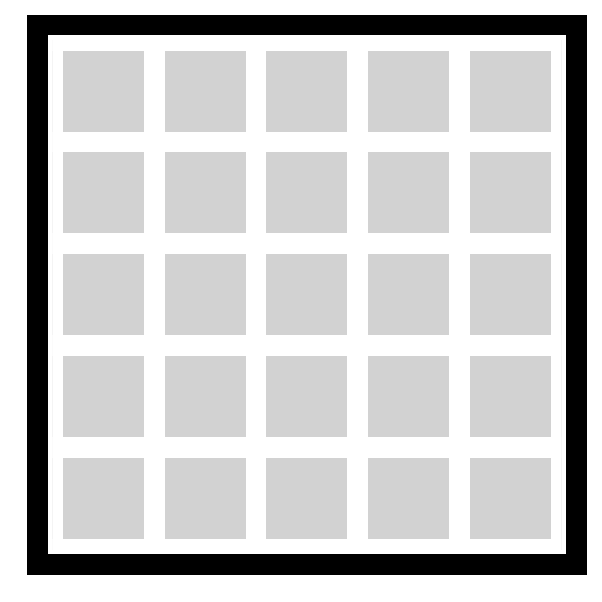}
            \end{minipage}
    }
    \subcaptionbox{$\alpha$=0.2}{
            \begin{minipage}{\columnwidth}
                \includegraphics[width=\textwidth]{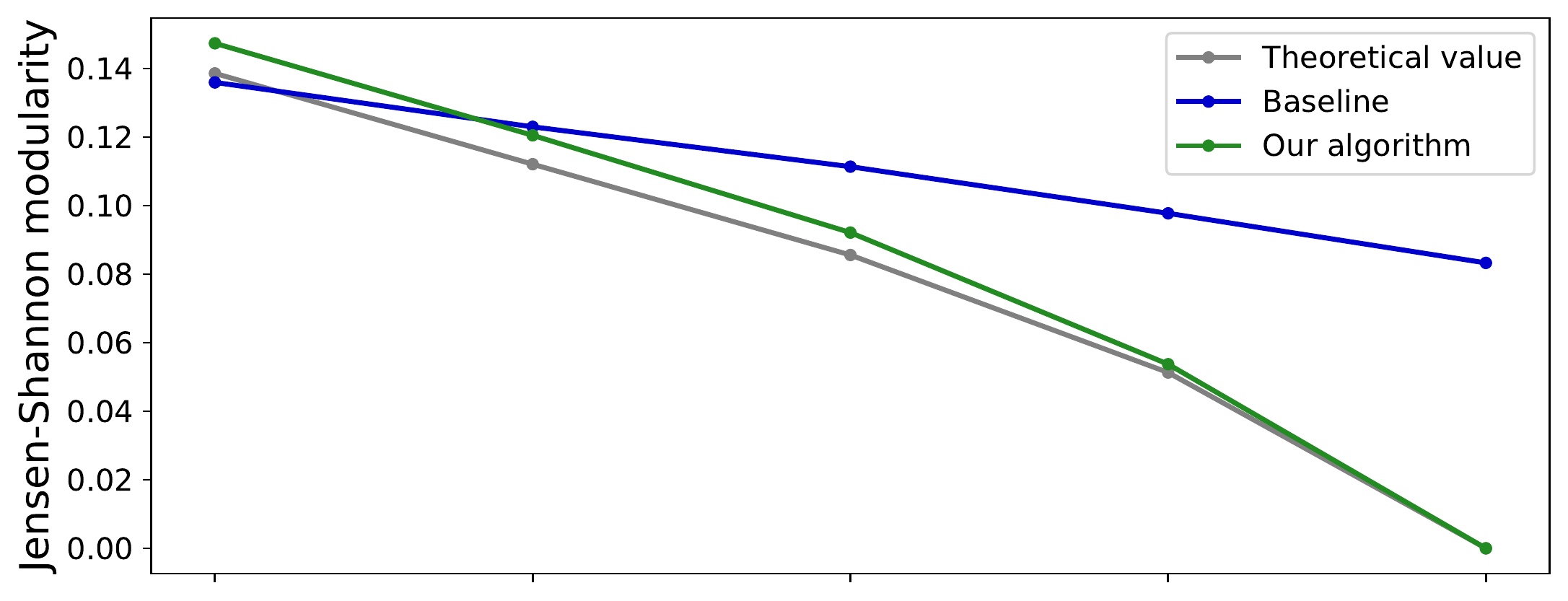}\\
                \null\hspace{0.045\textwidth}
                \includegraphics[width=0.18\textwidth]{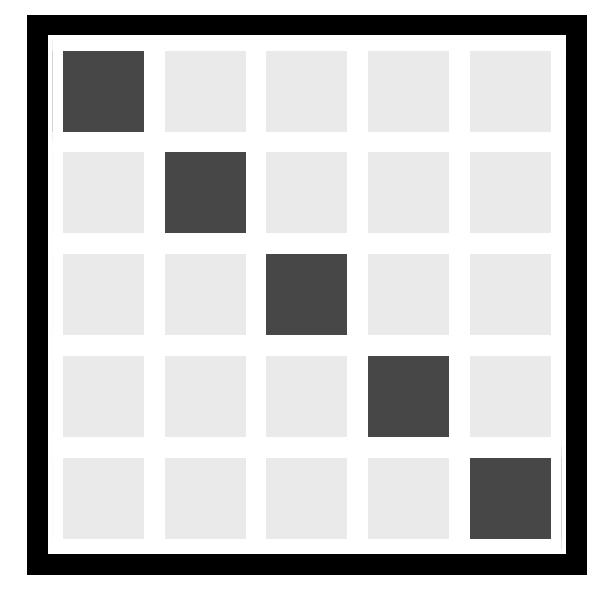}
                \includegraphics[width=0.18\textwidth]{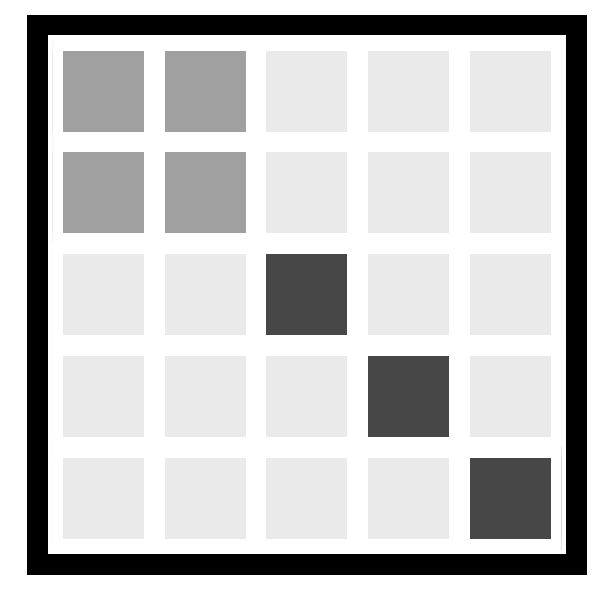}
                \includegraphics[width=0.18\textwidth]{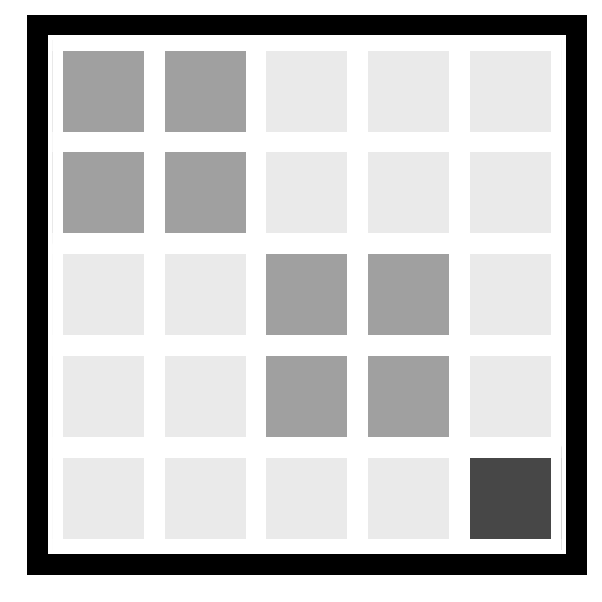}
                \includegraphics[width=0.18\textwidth]{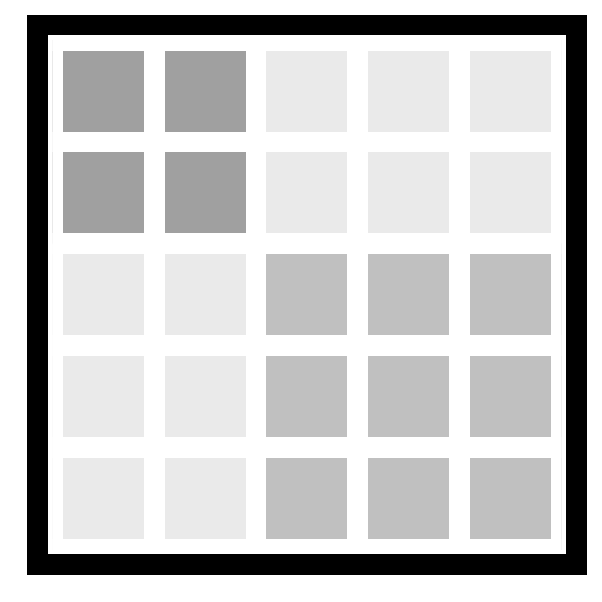}
                \includegraphics[width=0.18\textwidth]{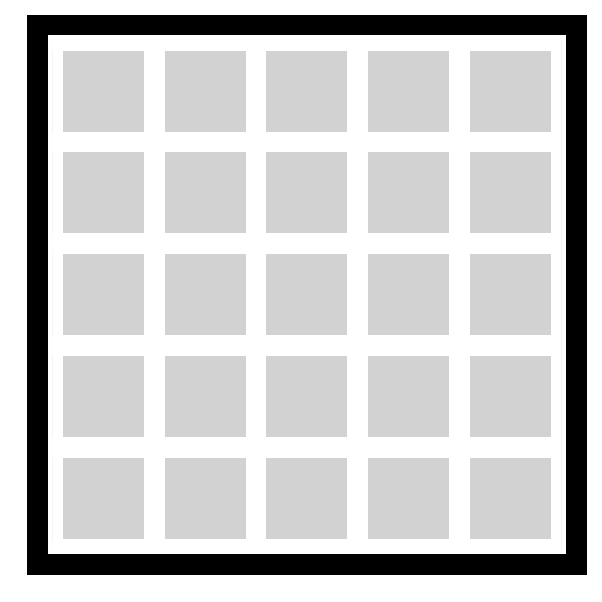}
            \end{minipage}
    }
    \caption{\textbf{Jensen-Shannon Modularity under Community Contraction.} We consider the bipartite graphs of size $200 \times 200$ with 40000 edges. The initial distribution has 5 identical communities and we contract the communities in the sequence of $(1)(2)(3)(4)(5)\rightarrow(12)(3)(4)(5)\rightarrow(12)(34)(5)\rightarrow(12)(345)\rightarrow(12345)$. The heat maps below illustrate the corresponding distribution matrices. The gray line is the theoretical value of Jensen-Shannon mutual information taking the true distribution as the input. The blue line shows the results of the baseline estimator, which calculates the mutual information directly by the noisy frequency matrix. The green line is estimated by our algorithm with $\theta = 0.9$ in the rank selection. As Jensen-Shannon divergence requires the non-negativity of $D_{u,v}$, we use NMF in the step of low-rank approximation. For each contraction stage, we independently generate 100 random graphs based on the underlying distribution and calculate the mean values. }
    \label{fig:jsmodularity-short}
\end{figure*}

\subsection{Community Contraction}

$f$-Modularity inherits the three properties of $f$-mutual information in Lemma \ref{properties}, symmetry, non-negativity and information monotonicity. The first two are obviously guaranteed by the definition of $f$-modularity, i.e.
\begin{itemize}
    \item $\modu^f(\left<U, V, E\right>) = \modu^f(\left<V, U, E\right>)$,
    \item $\modu^f(\left<U, V, E\right>) \geq 0$.
\end{itemize}
However, the last one, non-negativity, may not look straightforward.

Here we use community contraction to verity that $f$-modularity is approximately monotone regarding the level of community structure (see Figure~\ref{fig:block_example}). The information monotonicity of $f$-mutual information states that $\mi^f(T(X);Y)\leq \mi^f(X;Y)$ where $T(\cdot)\in \mathbb{R}^{|\Sigma_X|\times |\Sigma_X|}$ is a possibly random operator on $X$ whose randomness is independent of $Y$. Remark that any operator $T$ essentially multiplies a transition matrix to the joint distribution matrix. From the perspective of network, if we reduce the level of community structure by multiplying a transition matrix to the graph distribution, $f$-modularity should decrease. We choose community contraction as this operator for ease of presentation.

In detail, we contract two communities $(U_1,V_1),(U_2,V_2)$ by allocating the probability evenly within the merged community, i.e., for all $u\in U_1\cup U_2,v\in V_1\cup V_2$, we set the new probability of the edge $(u,v)$ to be 
\[
p'(u,v)=\frac{\sum_{s\in U_1\cup U_2,t\in V_1\cup V_2}p(s,t)}{|U_1\cup U_2||V_1\cup V_2|}.
\]

\subsection{Results}

Our model assumes the real network is a realization of the underlying distribution with noise. Detecting communities is essentially a process of eliminating noise and recover the true distribution. For example in Figure~\ref{fig:gen_example}, we aim to ``see'' the left matrix from the right matrix. So in addition to information monotonicity, we will also show the robustness of $f$-modularity for estimating the $f$-mutual information of the true joint distribution. While the \textit{theoretical value} of mutual information use the true distribution matrix, we set our \textit{baseline} estimator such that it directly calculates the mutual information by the noisy frequency matrix.

Figure~\ref{fig:jsmodularity-short} shows our numerical experiments on Jensen Shannon-modularity. We can see
\begin{itemize}
\item \textbf{Non-negativity}: Jensen-Shannon modularity is non-negative and approximately vanishes when the graph has no community structure;
\item \textbf{Information Monotonicity}: when we contract communities,  Jensen-Shannon modularity decreases;
\item \textbf{Robustness}: Compared to the baseline, Jensen-Shannon modularity provides a much more robust estimation for the theoretical value of Jensen-Shannon mutual information. 
\end{itemize}

More results for other instances of $f$-modularity are deferred to the appendix due to the limit of space. The above observations are also fit for them.



\section{Conclusion}
\label{conclusion}
In this paper, we propose a generalized modularity, $f$-modularity, based on the dual form of $f$-mutual information. We find a special case of TVD-modularity exactly matches Newman's modularity. We also give an algorithm that estimates $f$-modularity under the case of smooth functions $f$ and low-rank constraints $C$. Finally, we validate the properties of $f$-modularity by numerical experiments. Our work not only develops new measures for community structure, but also provides an information-theoretical interpretation to the concept of modularity. Modularity and mutual information, though lying in different areas, are two sides of the same coin.

So far we mainly focused on the low-rank constraint in this work for its simplicity. Our future work will explore not only better rank selection algorithms but also different constraints. Another interesting direction is to study networks with different topologies, like nested networks. We would like to employ information-theoretical tools to quantify more features in networks other than modularity.


    

\clearpage

\bibliographystyle{named}
\bibliography{reference.bib}


\clearpage

\onecolumn
\appendix

\section{Experimental Results on Jensen-Shannon Modularity}

\begin{figure}[H]
    \subcaptionbox{$\alpha$=0.1}{
        \begin{minipage}{0.5\textwidth}
            \includegraphics[width=\textwidth]{figures/charts_JS0.1.pdf}\\
                \null\hspace{0.045\textwidth}
                \includegraphics[width=0.18\textwidth]{mat/mat_1_0.pdf}
                \includegraphics[width=0.18\textwidth]{mat/mat_1_1.pdf}
                \includegraphics[width=0.18\textwidth]{mat/mat_1_2.pdf}
                \includegraphics[width=0.18\textwidth]{mat/mat_1_3.pdf}
                \includegraphics[width=0.18\textwidth]{mat/mat_1_4.pdf}
        \end{minipage}
    }
    \subcaptionbox{$\alpha$=0.2}{
        \begin{minipage}{0.5\textwidth}
            \includegraphics[width=\textwidth]{figures/charts_JS0.2.pdf}\\
                \null\hspace{0.045\textwidth}
                \includegraphics[width=0.18\textwidth]{mat/mat_2_0.pdf}
                \includegraphics[width=0.18\textwidth]{mat/mat_2_1.pdf}
                \includegraphics[width=0.18\textwidth]{mat/mat_2_2.pdf}
                \includegraphics[width=0.18\textwidth]{mat/mat_2_3.pdf}
                \includegraphics[width=0.18\textwidth]{mat/mat_2_4.pdf}
        \end{minipage}
    }
    \subcaptionbox{$\alpha$=0.3}{
        \begin{minipage}{0.5\textwidth}
            \includegraphics[width=\textwidth]{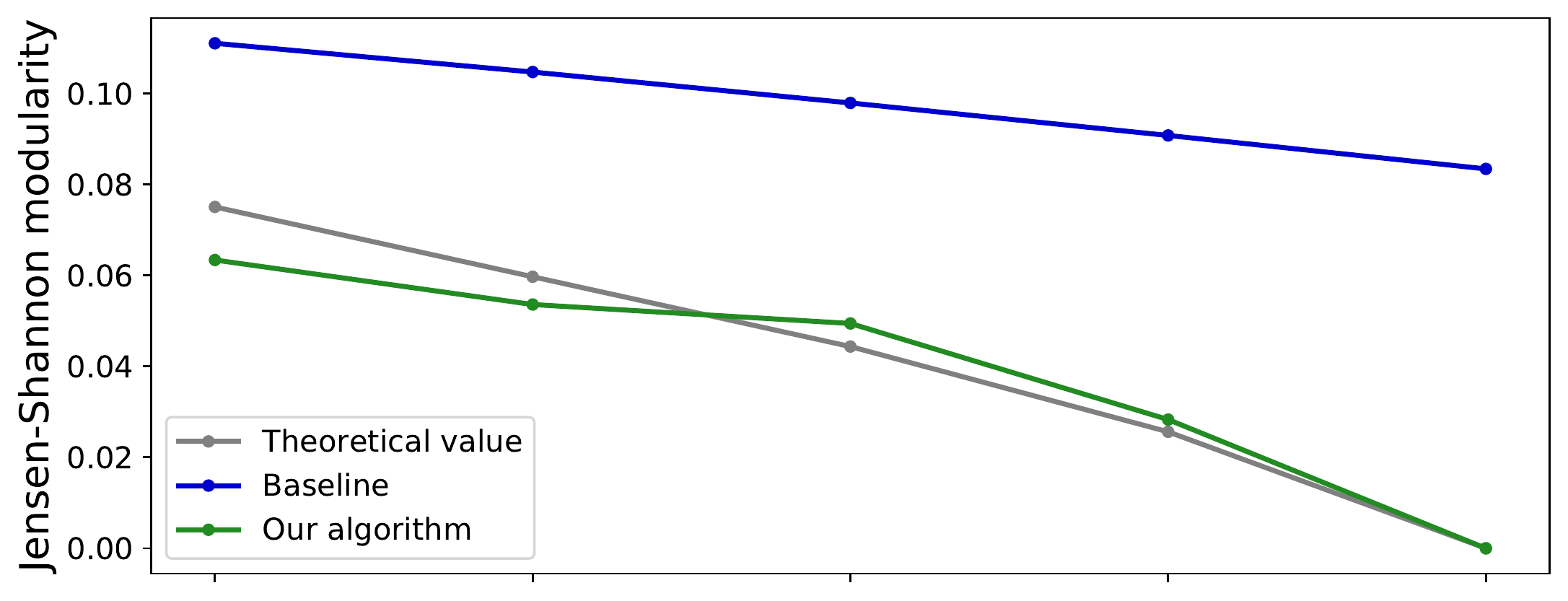}\\
                \null\hspace{0.045\textwidth}
                \includegraphics[width=0.18\textwidth]{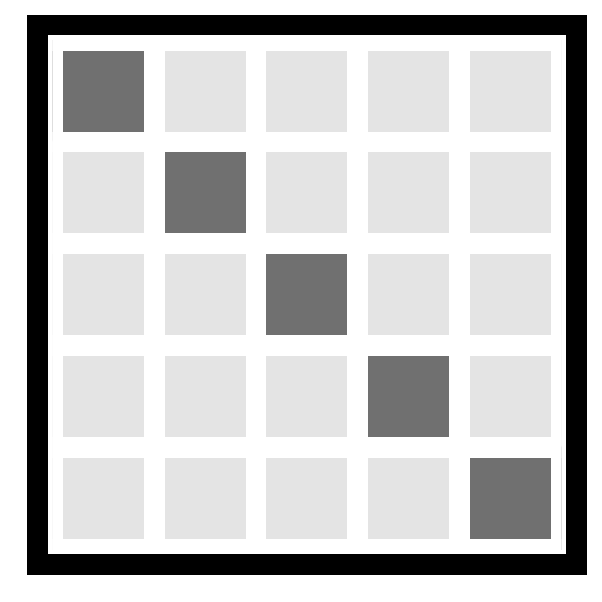}
                \includegraphics[width=0.18\textwidth]{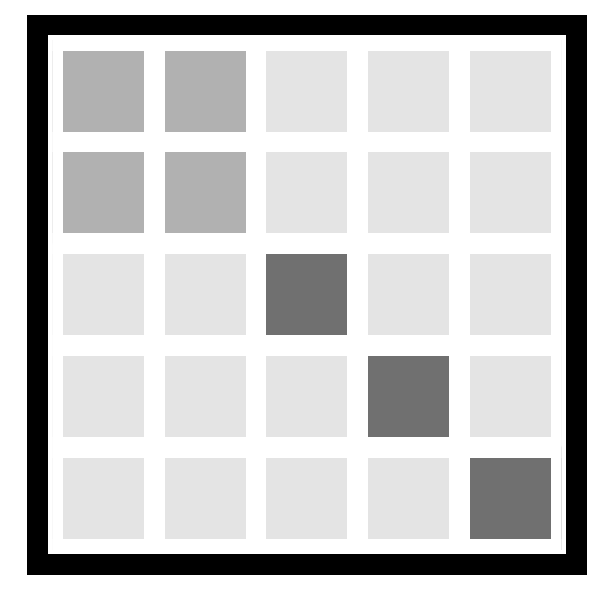}
                \includegraphics[width=0.18\textwidth]{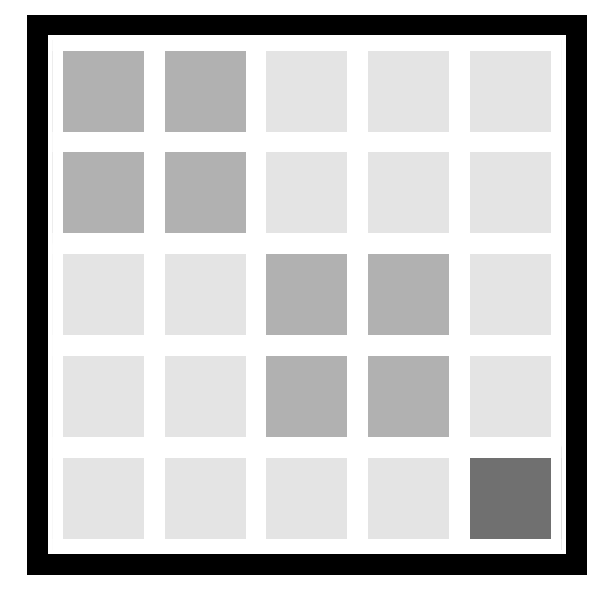}
                \includegraphics[width=0.18\textwidth]{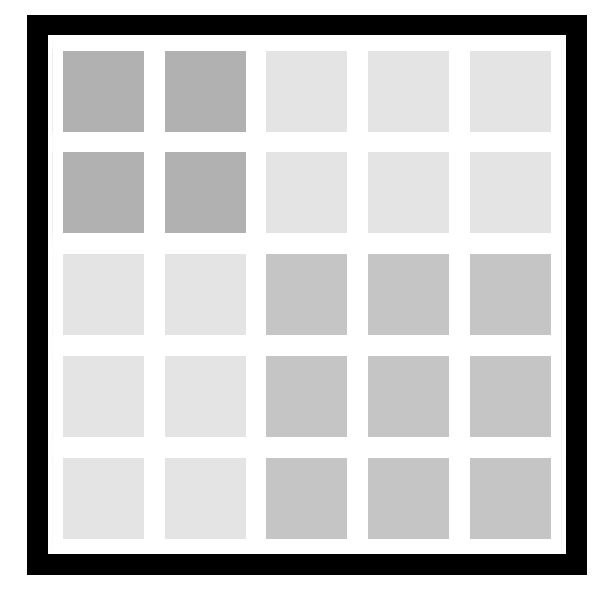}
                \includegraphics[width=0.18\textwidth]{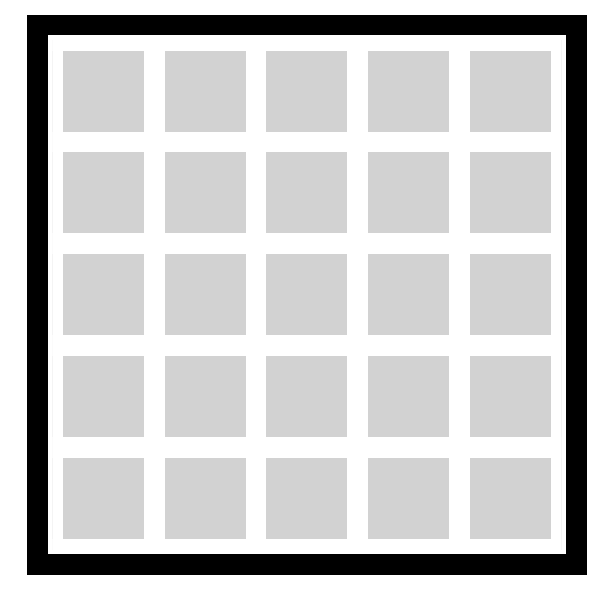}
        \end{minipage}
    }
    \subcaptionbox{$\alpha$=0.4}{
        \begin{minipage}{0.5\textwidth}
            \includegraphics[width=\textwidth]{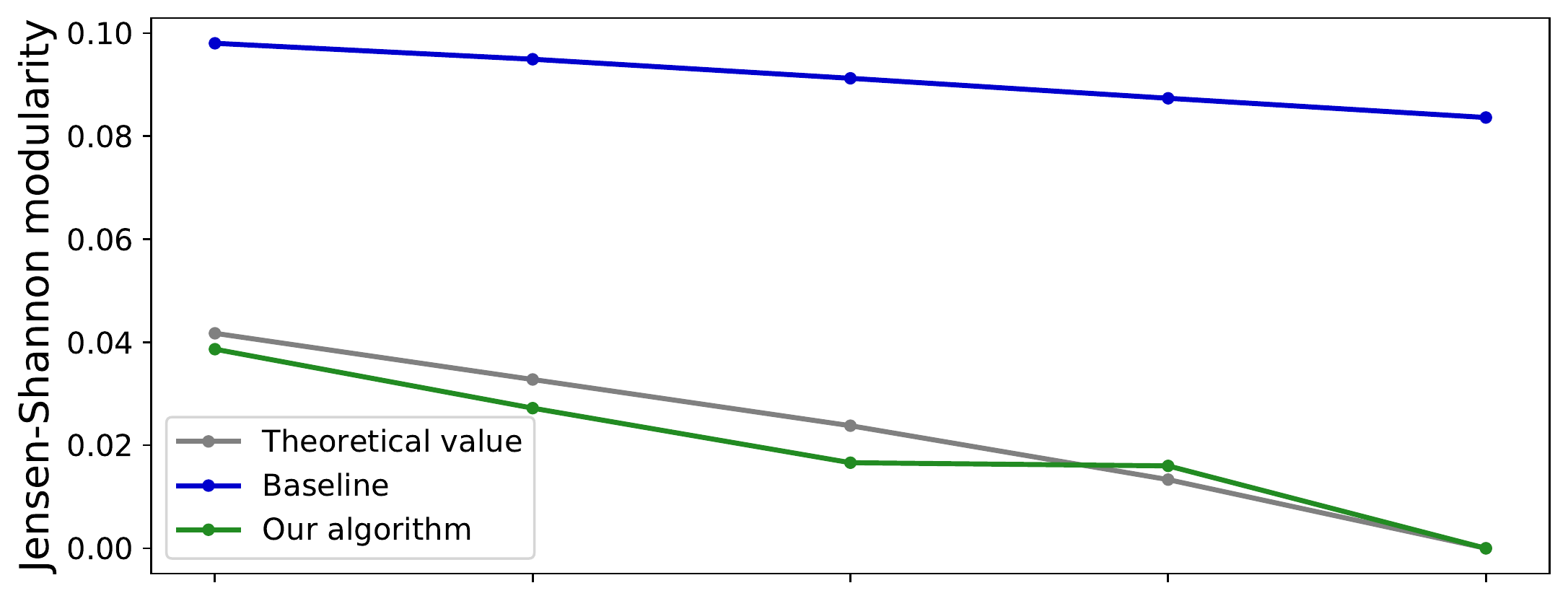}\\
                \null\hspace{0.045\textwidth}
                \includegraphics[width=0.18\textwidth]{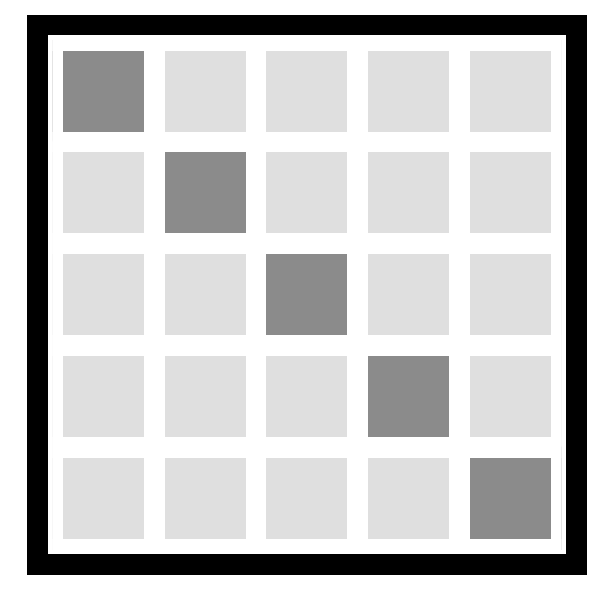}
                \includegraphics[width=0.18\textwidth]{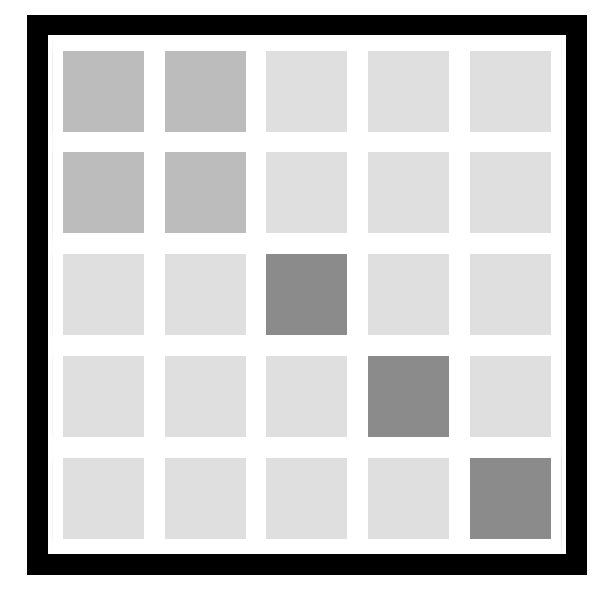}
                \includegraphics[width=0.18\textwidth]{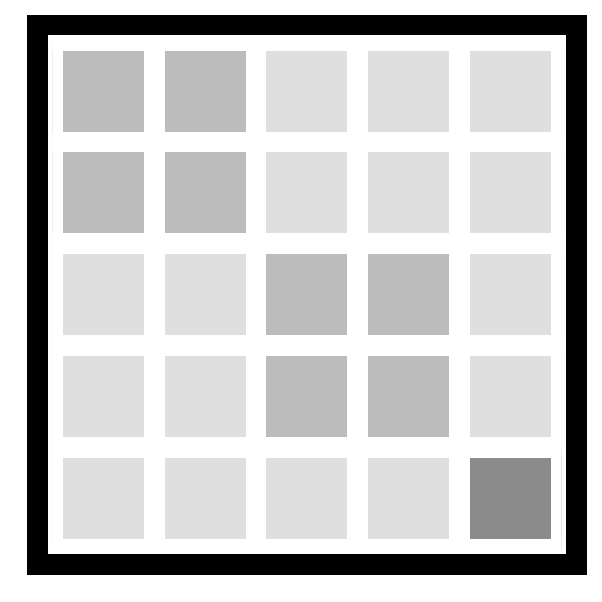}
                \includegraphics[width=0.18\textwidth]{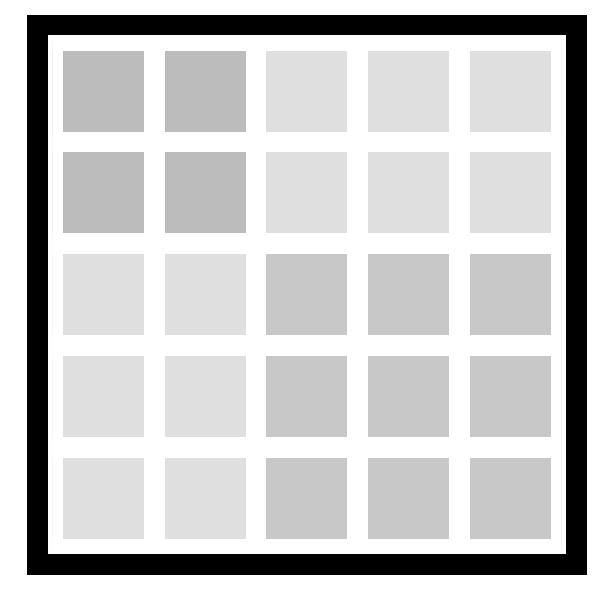}
                \includegraphics[width=0.18\textwidth]{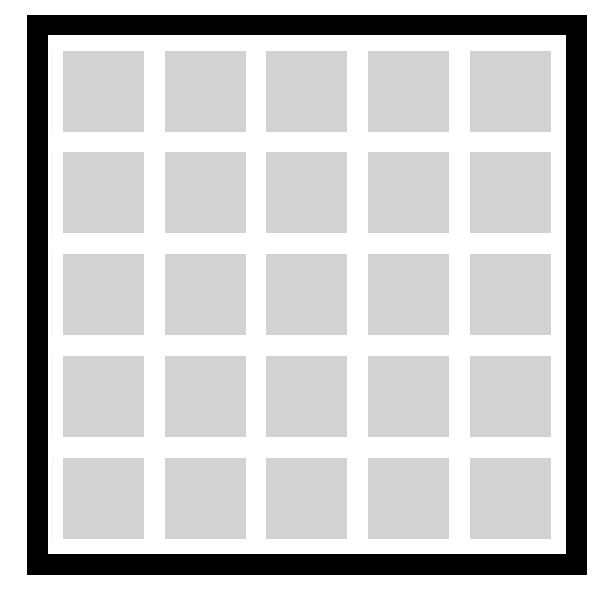}
        \end{minipage}
    }
    \subcaptionbox{$\alpha$=0.5}{
        \begin{minipage}{0.5\textwidth}
            \includegraphics[width=\textwidth]{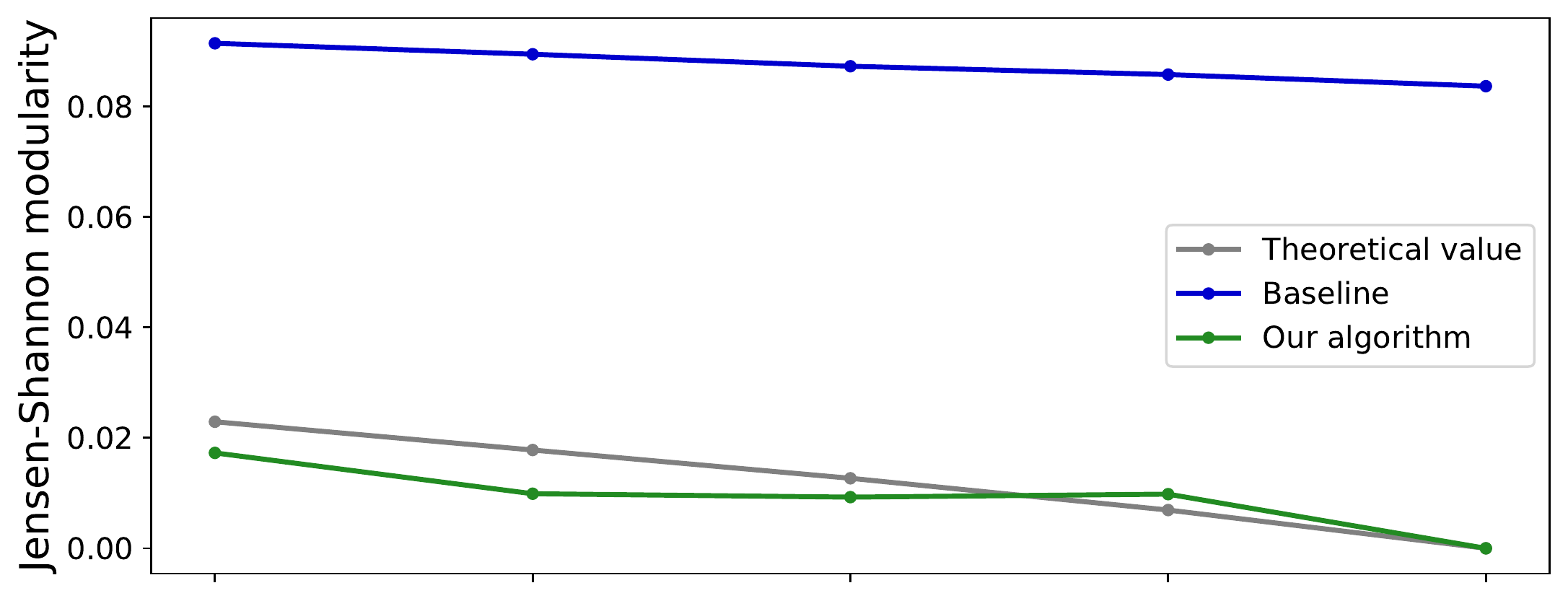}\\
                \null\hspace{0.045\textwidth}
                \includegraphics[width=0.18\textwidth]{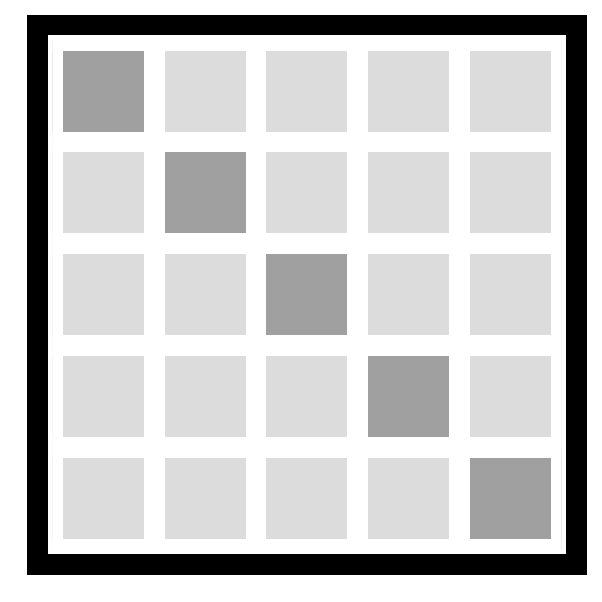}
                \includegraphics[width=0.18\textwidth]{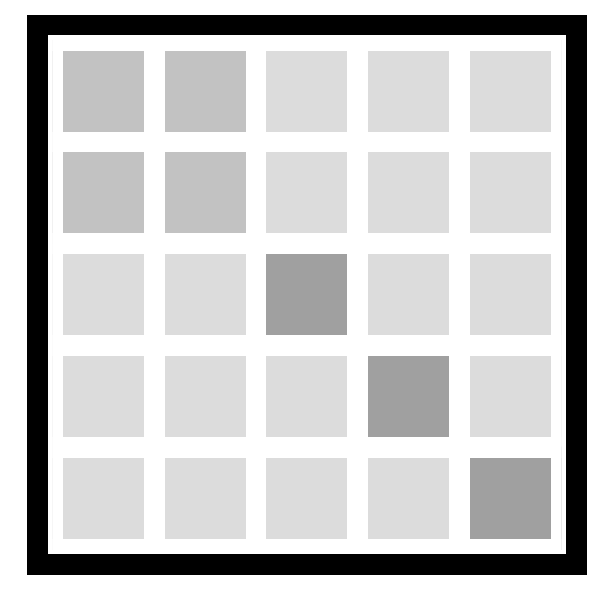}
                \includegraphics[width=0.18\textwidth]{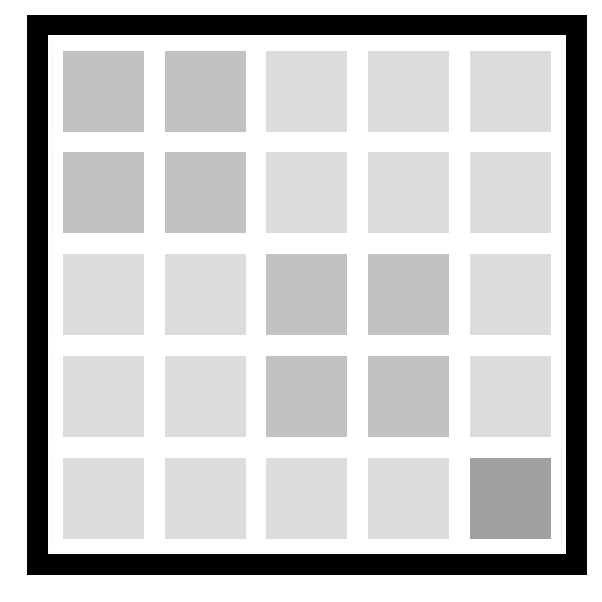}
                \includegraphics[width=0.18\textwidth]{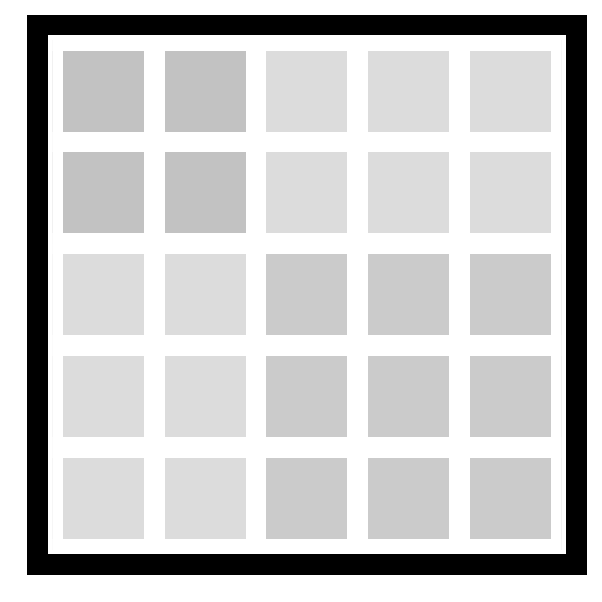}
                \includegraphics[width=0.18\textwidth]{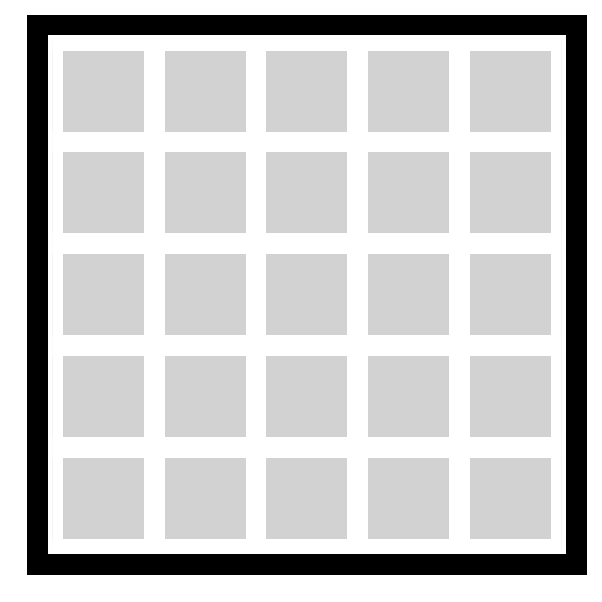}
        \end{minipage}
    }
    \subcaptionbox{$\alpha$=0.6}{
        \begin{minipage}{0.5\textwidth}
            \includegraphics[width=\textwidth]{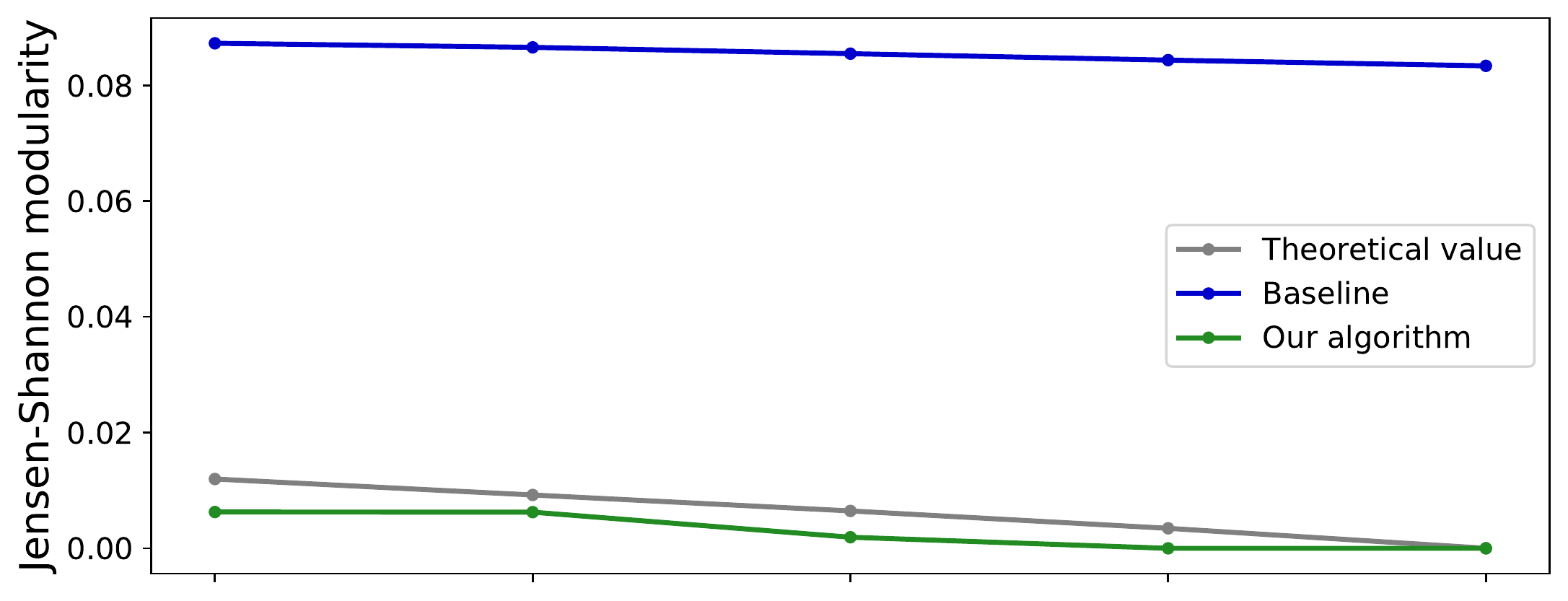}\\
                \null\hspace{0.045\textwidth}
                \includegraphics[width=0.18\textwidth]{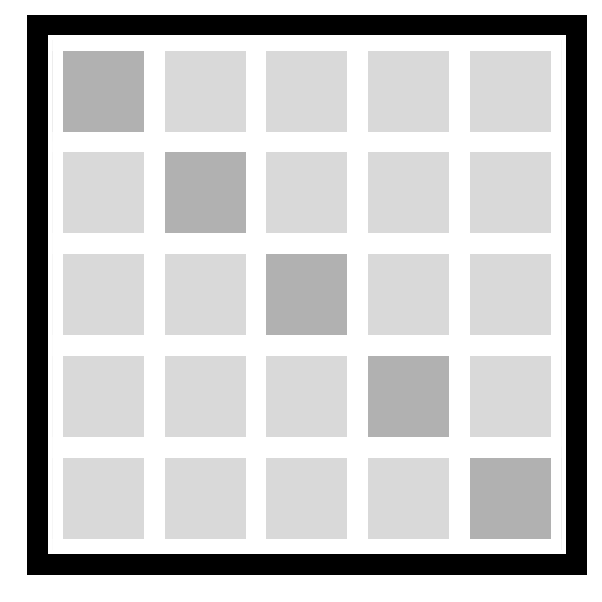}
                \includegraphics[width=0.18\textwidth]{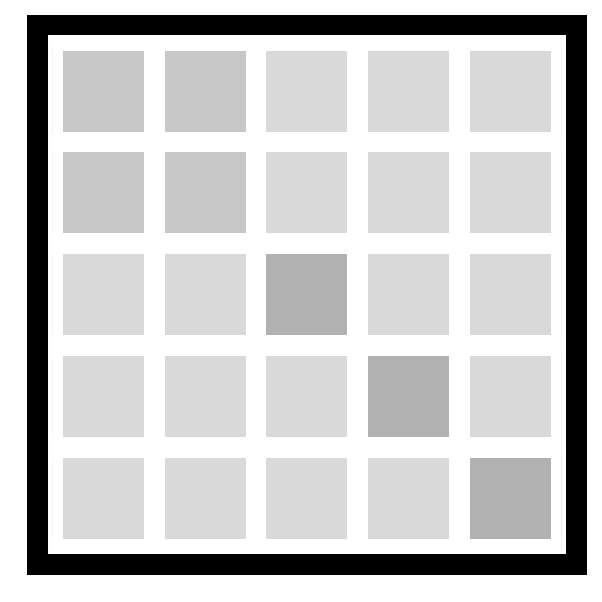}
                \includegraphics[width=0.18\textwidth]{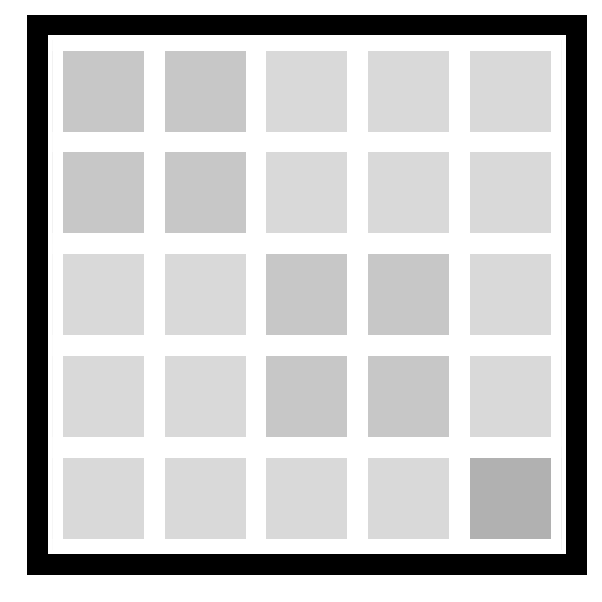}
                \includegraphics[width=0.18\textwidth]{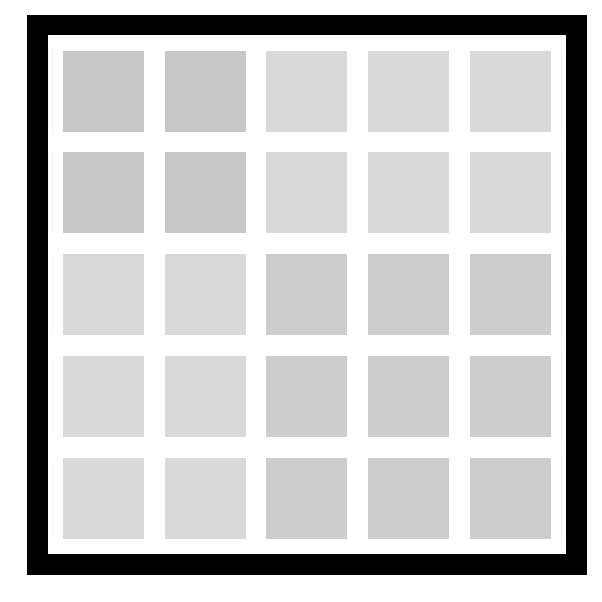}
                \includegraphics[width=0.18\textwidth]{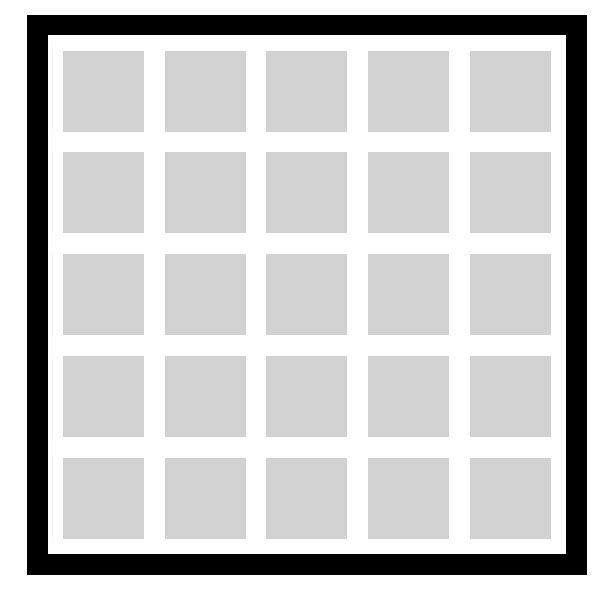}
        \end{minipage}
    }
    \caption{\textbf{Jensen-Shannon Modularity under Community Contraction.} The settings are the same as in Figure~\ref{fig:jsmodularity-short}.}
    \label{fig:jsmodularity}
\end{figure}

\newpage

\section{Experimental Results on KL-Modularity}

\begin{figure}[h!]
    \subcaptionbox{$\alpha$=0.1}{
        \begin{minipage}{0.5\textwidth}
            \includegraphics[width=\textwidth]{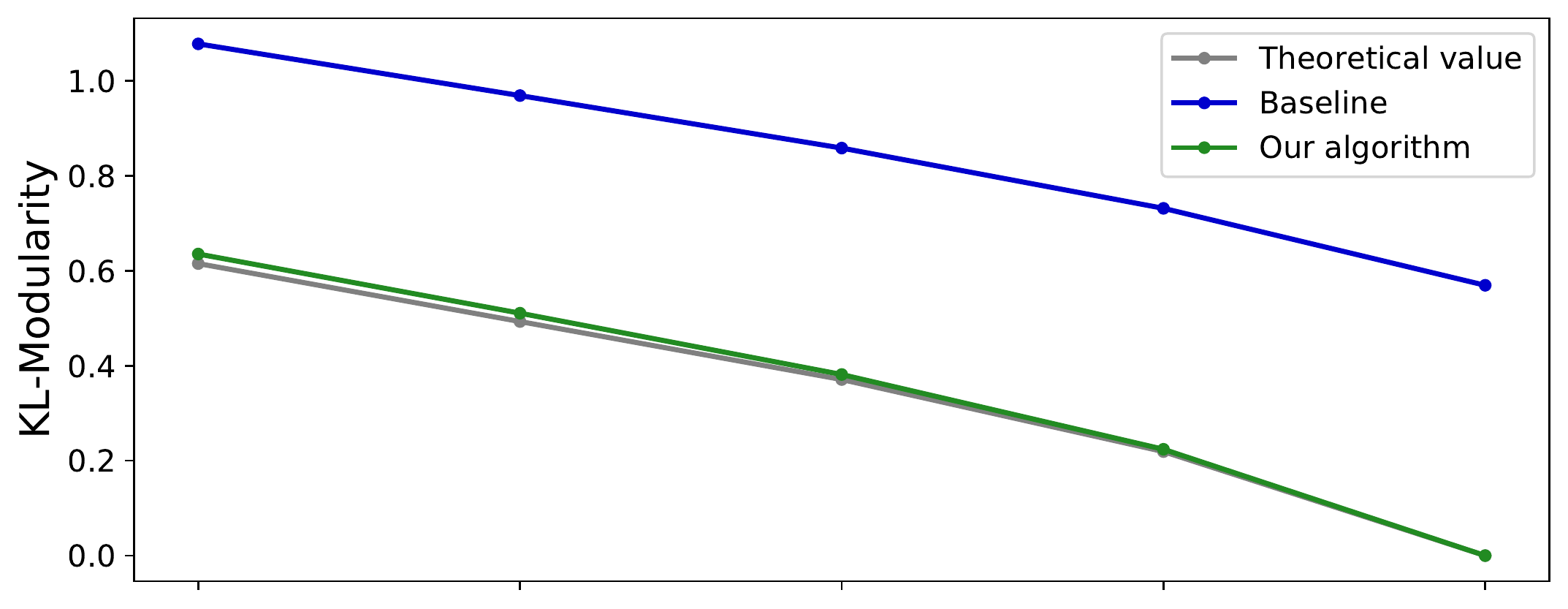}\\
                \null\hspace{0.045\textwidth}
                \includegraphics[width=0.18\textwidth]{mat/mat_1_0.pdf}
                \includegraphics[width=0.18\textwidth]{mat/mat_1_1.pdf}
                \includegraphics[width=0.18\textwidth]{mat/mat_1_2.pdf}
                \includegraphics[width=0.18\textwidth]{mat/mat_1_3.pdf}
                \includegraphics[width=0.18\textwidth]{mat/mat_1_4.pdf}
        \end{minipage}
    }
    \subcaptionbox{$\alpha$=0.2}{
        \begin{minipage}{0.5\textwidth}
            \includegraphics[width=\textwidth]{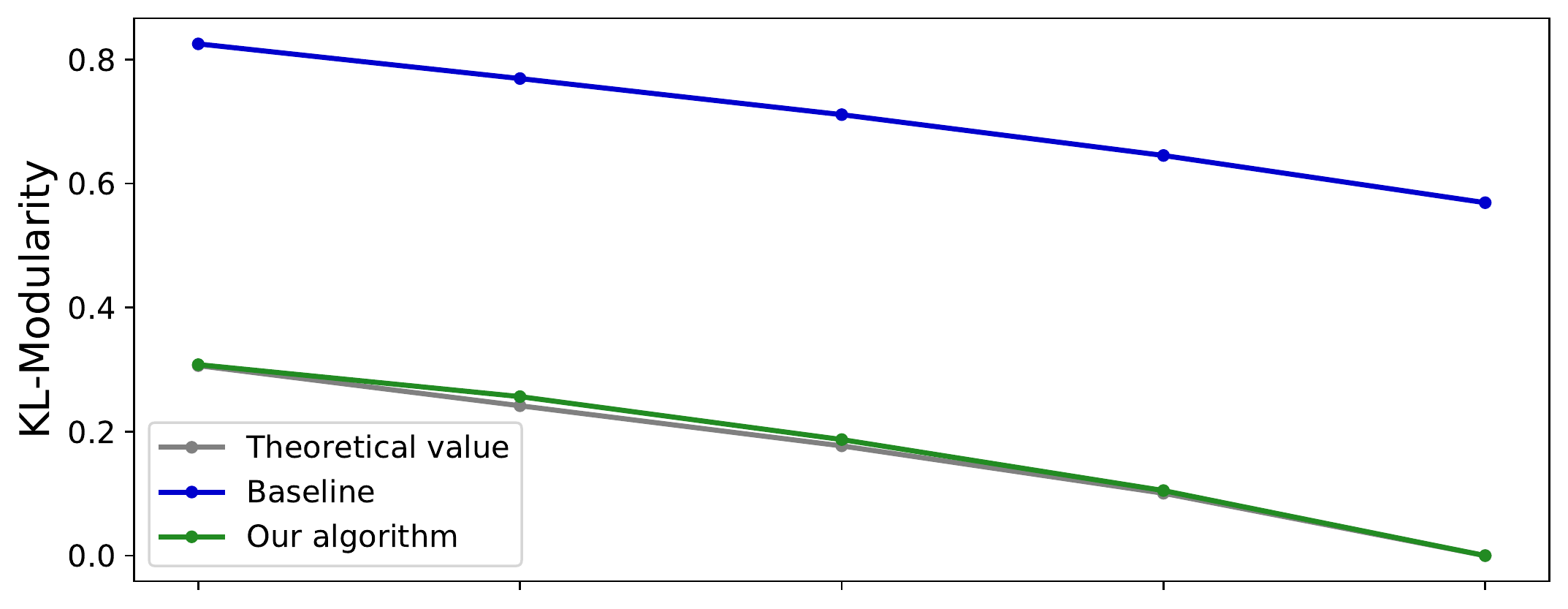}\\
                \null\hspace{0.045\textwidth}
                \includegraphics[width=0.18\textwidth]{mat/mat_2_0.pdf}
                \includegraphics[width=0.18\textwidth]{mat/mat_2_1.pdf}
                \includegraphics[width=0.18\textwidth]{mat/mat_2_2.pdf}
                \includegraphics[width=0.18\textwidth]{mat/mat_2_3.pdf}
                \includegraphics[width=0.18\textwidth]{mat/mat_2_4.pdf}
        \end{minipage}
    }
    \subcaptionbox{$\alpha$=0.3}{
        \begin{minipage}{0.5\textwidth}
            \includegraphics[width=\textwidth]{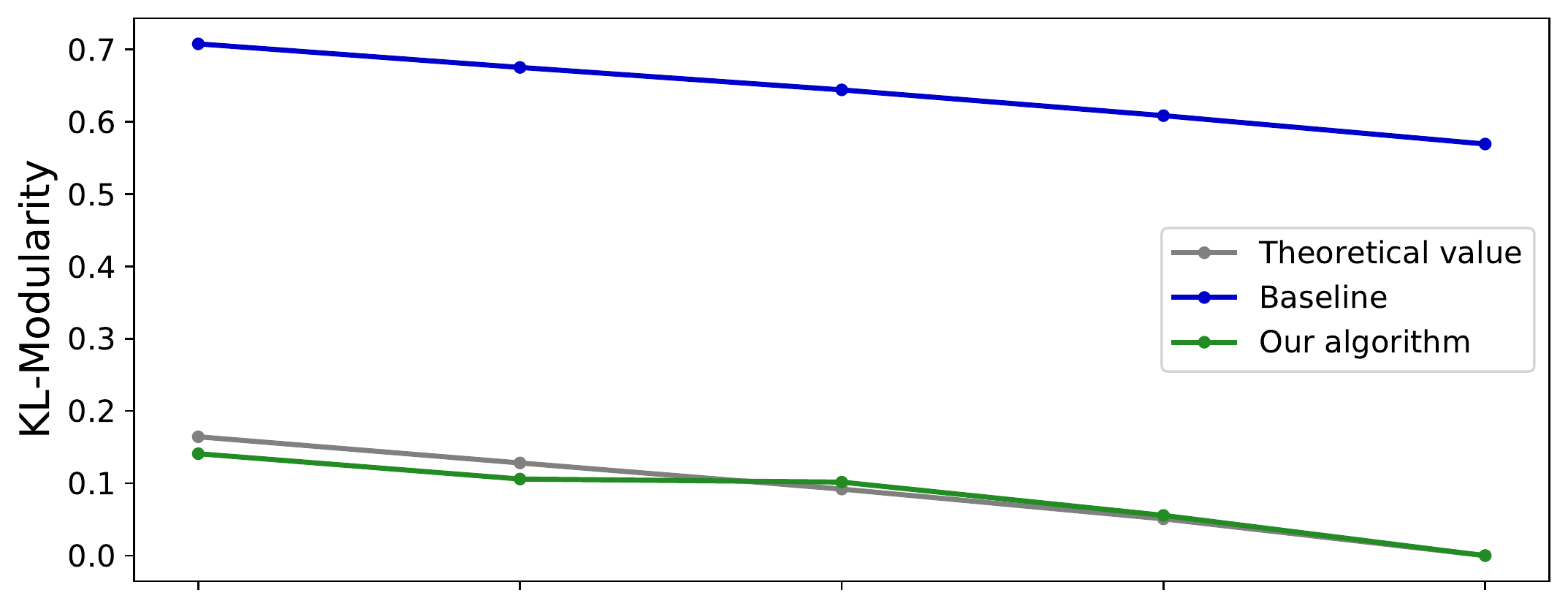}\\
                \null\hspace{0.045\textwidth}
                \includegraphics[width=0.18\textwidth]{mat/mat_3_0.pdf}
                \includegraphics[width=0.18\textwidth]{mat/mat_3_1.pdf}
                \includegraphics[width=0.18\textwidth]{mat/mat_3_2.pdf}
                \includegraphics[width=0.18\textwidth]{mat/mat_3_3.pdf}
                \includegraphics[width=0.18\textwidth]{mat/mat_3_4.pdf}
        \end{minipage}
    }
    \subcaptionbox{$\alpha$=0.4}{
        \begin{minipage}{0.5\textwidth}
            \includegraphics[width=\textwidth]{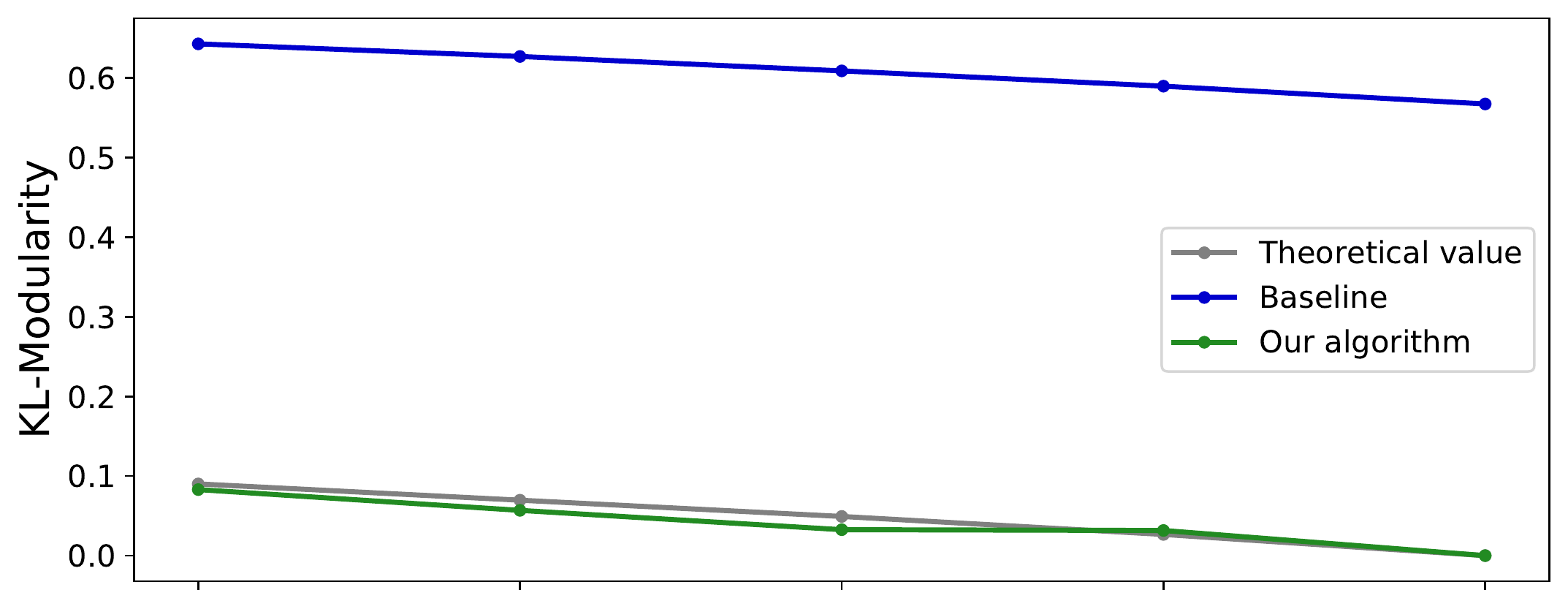}\\
                \null\hspace{0.045\textwidth}
                \includegraphics[width=0.18\textwidth]{mat/mat_4_0.pdf}
                \includegraphics[width=0.18\textwidth]{mat/mat_4_1.pdf}
                \includegraphics[width=0.18\textwidth]{mat/mat_4_2.pdf}
                \includegraphics[width=0.18\textwidth]{mat/mat_4_3.pdf}
                \includegraphics[width=0.18\textwidth]{mat/mat_4_4.pdf}
        \end{minipage}
    }
    \subcaptionbox{$\alpha$=0.5}{
        \begin{minipage}{0.5\textwidth}
            \includegraphics[width=\textwidth]{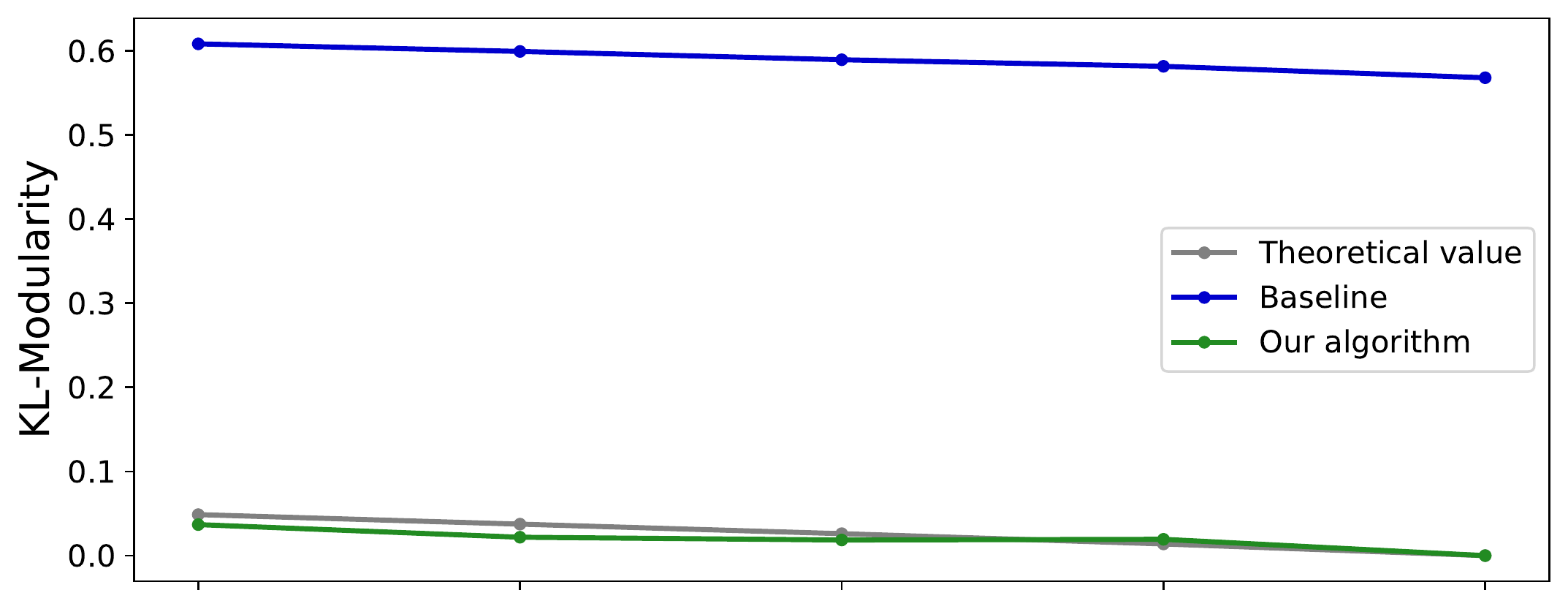}\\
                \null\hspace{0.045\textwidth}
                \includegraphics[width=0.18\textwidth]{mat/mat_5_0.pdf}
                \includegraphics[width=0.18\textwidth]{mat/mat_5_1.pdf}
                \includegraphics[width=0.18\textwidth]{mat/mat_5_2.pdf}
                \includegraphics[width=0.18\textwidth]{mat/mat_5_3.pdf}
                \includegraphics[width=0.18\textwidth]{mat/mat_5_4.pdf}
        \end{minipage}
    }
    \subcaptionbox{$\alpha$=0.6}{
        \begin{minipage}{0.5\textwidth}
            \includegraphics[width=\textwidth]{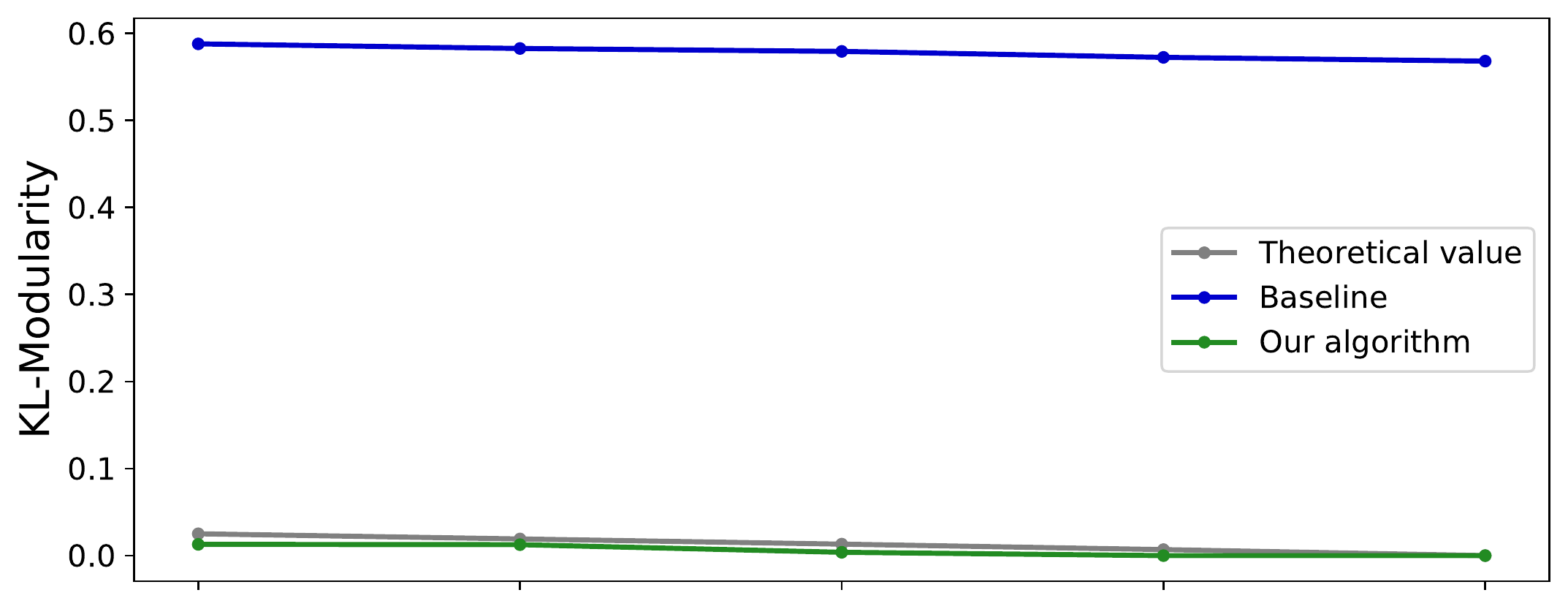}\\
                \null\hspace{0.045\textwidth}
                \includegraphics[width=0.18\textwidth]{mat/mat_6_0.pdf}
                \includegraphics[width=0.18\textwidth]{mat/mat_6_1.pdf}
                \includegraphics[width=0.18\textwidth]{mat/mat_6_2.pdf}
                \includegraphics[width=0.18\textwidth]{mat/mat_6_3.pdf}
                \includegraphics[width=0.18\textwidth]{mat/mat_6_4.pdf}
        \end{minipage}
    }
    \caption{\textbf{KL-Modularity under Community Contraction.} The corresponding KL-mutual information is Shannon mutual information. Other settings are the same as in Figure~\ref{fig:jsmodularity-short}.}
    \label{fig:klmodularity}
\end{figure}

\newpage

\section{Experimental Results on Pearson-Modularity}

\begin{figure}[h!]
    \subcaptionbox{$\alpha$=0.1}{
        \begin{minipage}{0.5\textwidth}
            \includegraphics[width=\textwidth]{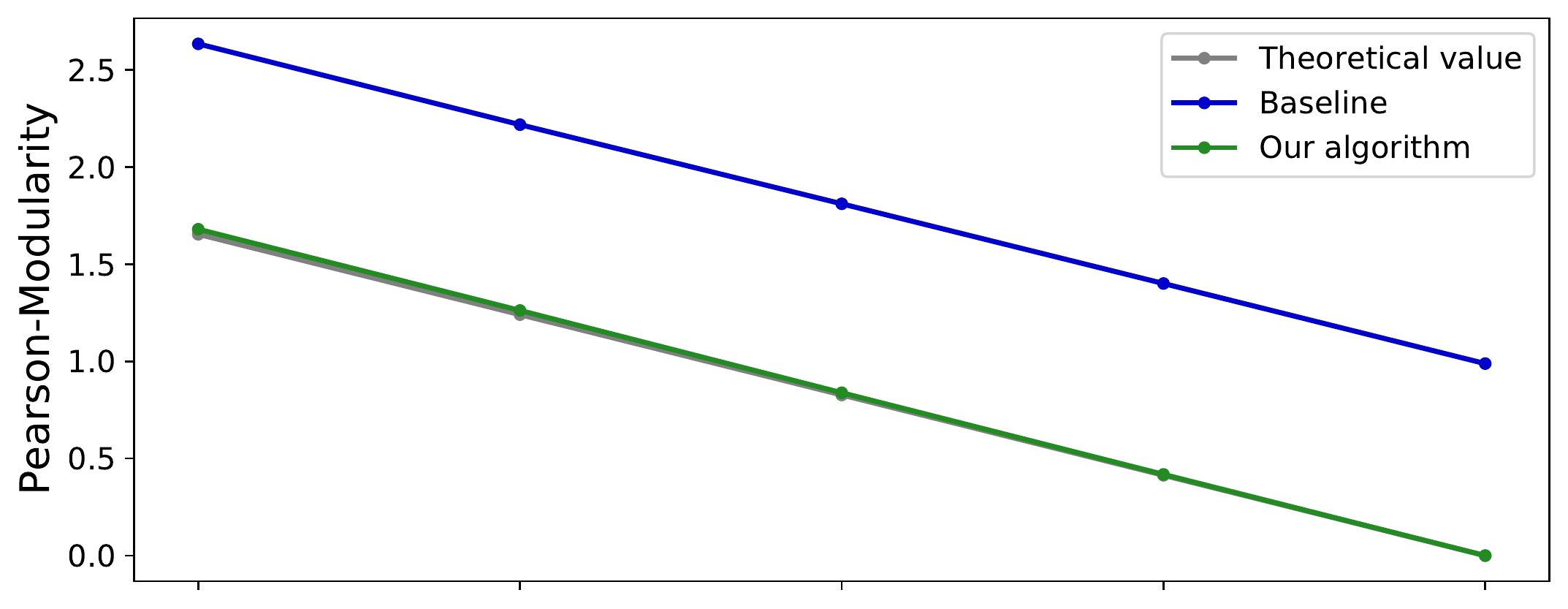}\\
                \null\hspace{0.045\textwidth}
                \includegraphics[width=0.18\textwidth]{mat/mat_1_0.pdf}
                \includegraphics[width=0.18\textwidth]{mat/mat_1_1.pdf}
                \includegraphics[width=0.18\textwidth]{mat/mat_1_2.pdf}
                \includegraphics[width=0.18\textwidth]{mat/mat_1_3.pdf}
                \includegraphics[width=0.18\textwidth]{mat/mat_1_4.pdf}
        \end{minipage}
    }
    \subcaptionbox{$\alpha$=0.2}{
        \begin{minipage}{0.5\textwidth}
            \includegraphics[width=\textwidth]{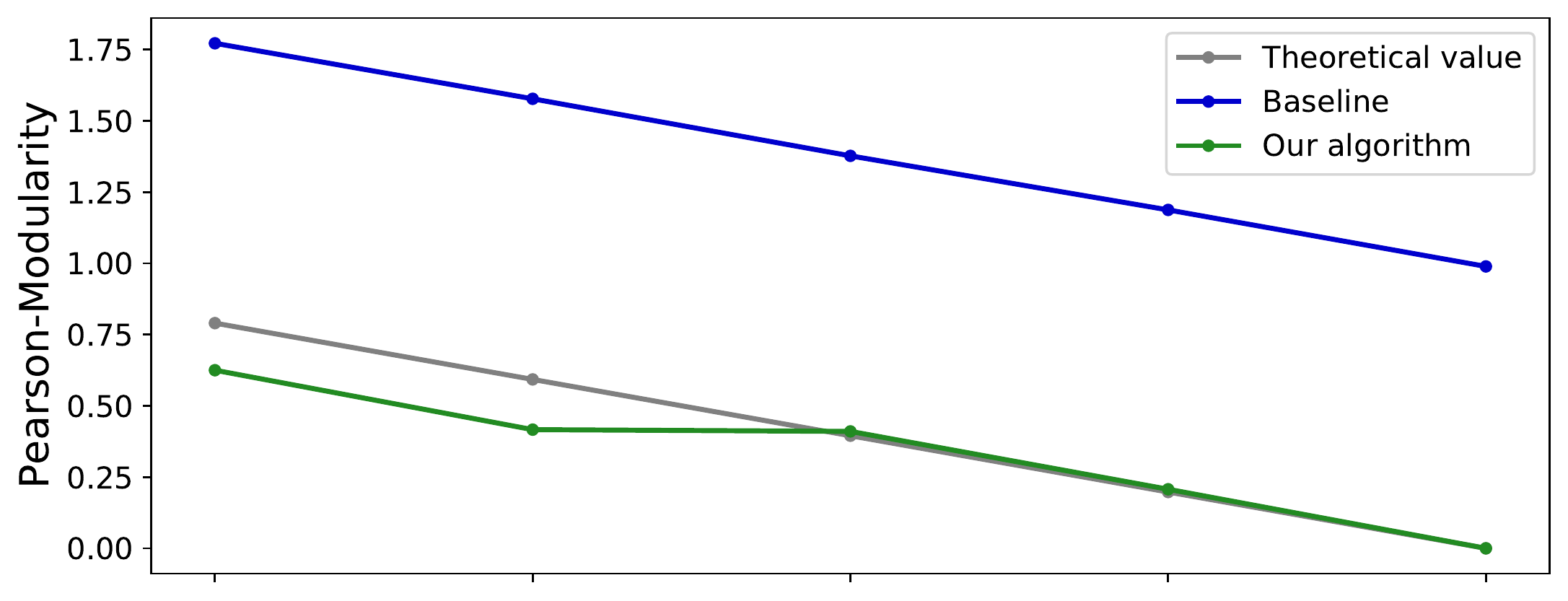}\\
                \null\hspace{0.045\textwidth}
                \includegraphics[width=0.18\textwidth]{mat/mat_2_0.pdf}
                \includegraphics[width=0.18\textwidth]{mat/mat_2_1.pdf}
                \includegraphics[width=0.18\textwidth]{mat/mat_2_2.pdf}
                \includegraphics[width=0.18\textwidth]{mat/mat_2_3.pdf}
                \includegraphics[width=0.18\textwidth]{mat/mat_2_4.pdf}
        \end{minipage}
    }
    \subcaptionbox{$\alpha$=0.3}{
        \begin{minipage}{0.5\textwidth}
            \includegraphics[width=\textwidth]{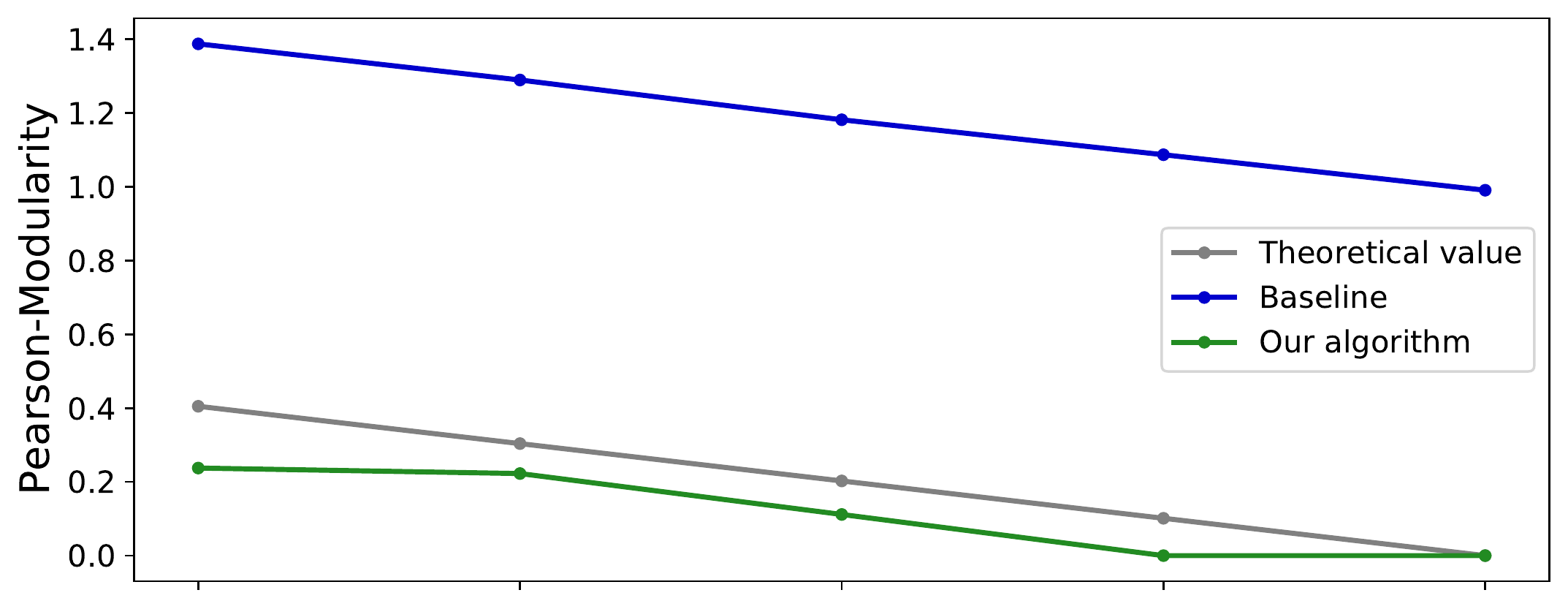}\\
                \null\hspace{0.045\textwidth}
                \includegraphics[width=0.18\textwidth]{mat/mat_3_0.pdf}
                \includegraphics[width=0.18\textwidth]{mat/mat_3_1.pdf}
                \includegraphics[width=0.18\textwidth]{mat/mat_3_2.pdf}
                \includegraphics[width=0.18\textwidth]{mat/mat_3_3.pdf}
                \includegraphics[width=0.18\textwidth]{mat/mat_3_4.pdf}
        \end{minipage}
    }
    \subcaptionbox{$\alpha$=0.4}{
        \begin{minipage}{0.5\textwidth}
            \includegraphics[width=\textwidth]{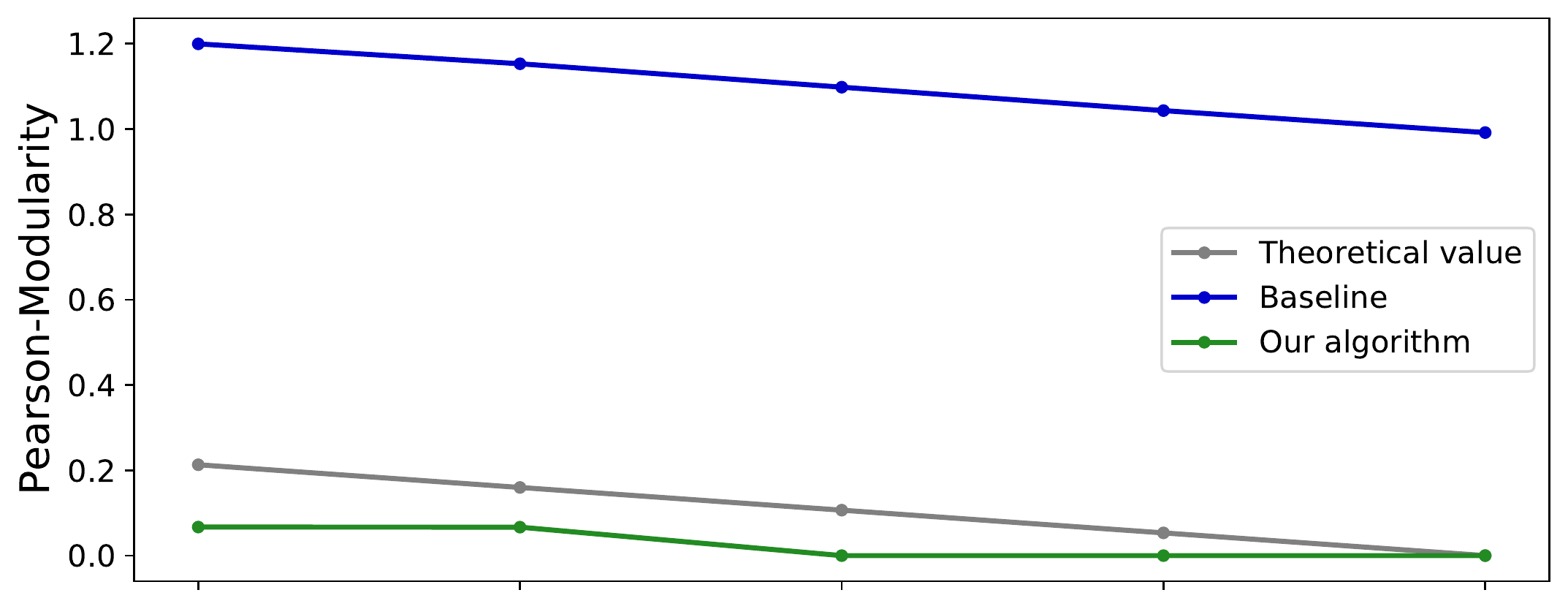}\\
                \null\hspace{0.045\textwidth}
                \includegraphics[width=0.18\textwidth]{mat/mat_4_0.pdf}
                \includegraphics[width=0.18\textwidth]{mat/mat_4_1.pdf}
                \includegraphics[width=0.18\textwidth]{mat/mat_4_2.pdf}
                \includegraphics[width=0.18\textwidth]{mat/mat_4_3.pdf}
                \includegraphics[width=0.18\textwidth]{mat/mat_4_4.pdf}
        \end{minipage}
    }
    \subcaptionbox{$\alpha$=0.5}{
        \begin{minipage}{0.5\textwidth}
            \includegraphics[width=\textwidth]{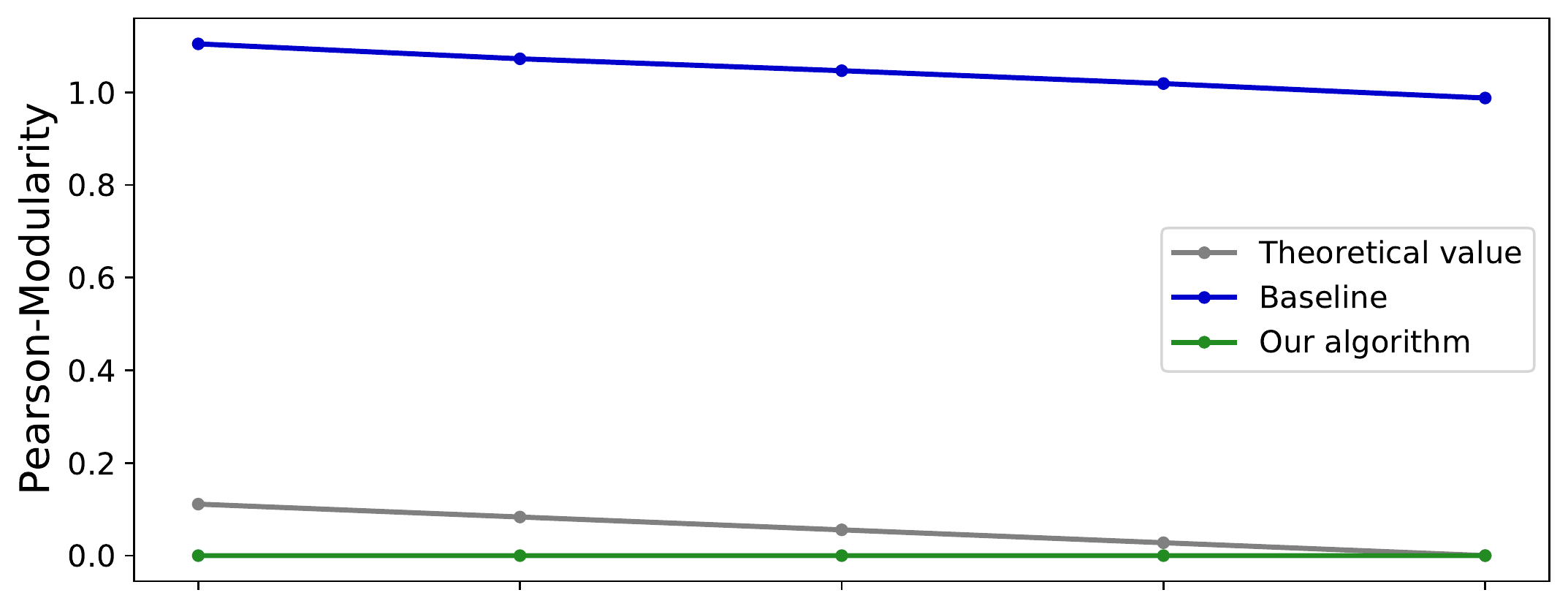}\\
                \null\hspace{0.045\textwidth}
                \includegraphics[width=0.18\textwidth]{mat/mat_5_0.pdf}
                \includegraphics[width=0.18\textwidth]{mat/mat_5_1.pdf}
                \includegraphics[width=0.18\textwidth]{mat/mat_5_2.pdf}
                \includegraphics[width=0.18\textwidth]{mat/mat_5_3.pdf}
                \includegraphics[width=0.18\textwidth]{mat/mat_5_4.pdf}
        \end{minipage}
    }
    \subcaptionbox{$\alpha$=0.6}{
        \begin{minipage}{0.5\textwidth}
            \includegraphics[width=\textwidth]{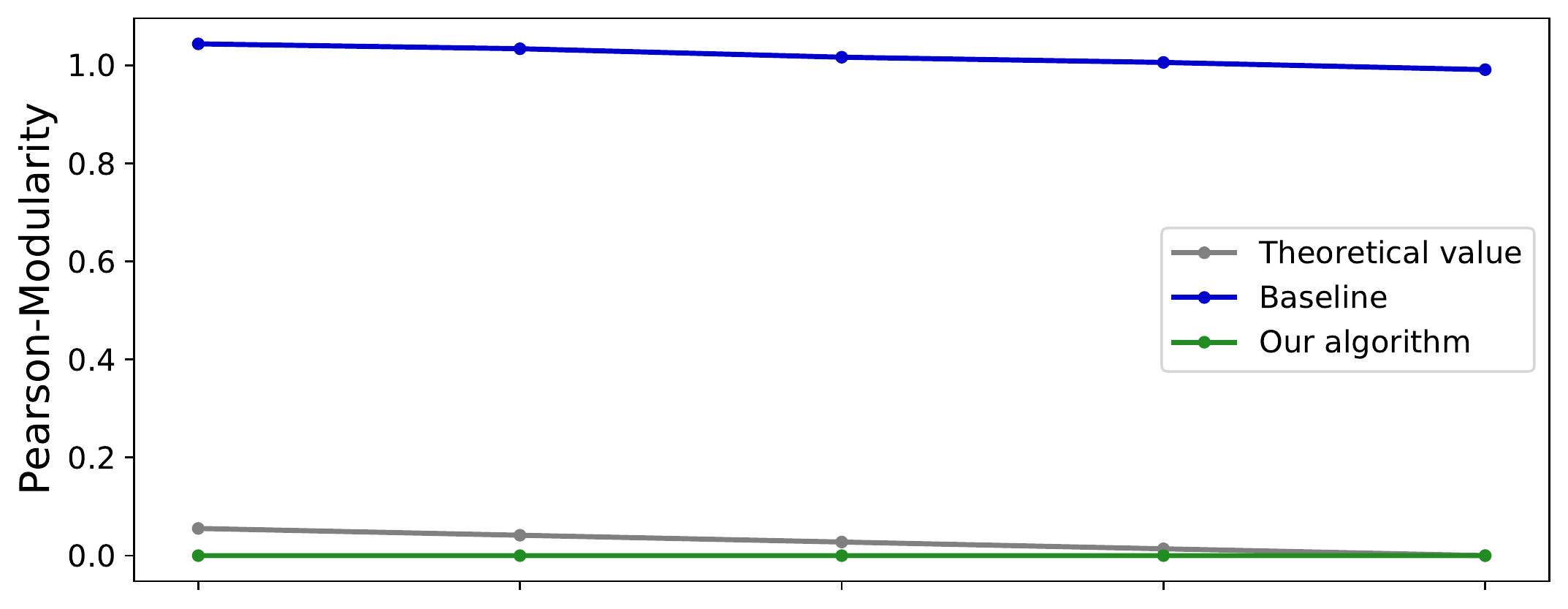}\\
                \null\hspace{0.045\textwidth}
                \includegraphics[width=0.18\textwidth]{mat/mat_6_0.pdf}
                \includegraphics[width=0.18\textwidth]{mat/mat_6_1.pdf}
                \includegraphics[width=0.18\textwidth]{mat/mat_6_2.pdf}
                \includegraphics[width=0.18\textwidth]{mat/mat_6_3.pdf}
                \includegraphics[width=0.18\textwidth]{mat/mat_6_4.pdf}
        \end{minipage}
    }
    \caption{\textbf{Pearson-Modularity under Community Contraction.} As Pearson-divergence does not require the non-negativity of $D_{u,v}$, we can use SVD for the low-rank approximation. Other settings are the same as in Figure~\ref{fig:jsmodularity-short}.}
    \label{fig:pearsonmodularity}
\end{figure}

\end{document}